\documentclass[sigconf,natbib=true,anonymous=false]{acmart}

\AtBeginDocument{%
  \providecommand\BibTeX{{%
    \normalfont B\kern-0.5em{\scshape i\kern-0.25em b}\kern-0.8em\TeX}}}

\setcopyright{acmlicensed}
\copyrightyear{2024}

\acmConference[WSDM '24]{Proceedings of The 17th ACM International Conference on Web Search and Data Mining}{March 4--8, 2024}{Mérida, Mexico}
\acmBooktitle{Proceedings of the 17th ACM International Conference on Web
Search and Data Mining (WSDM '24), March 4--8, 2024, Merida, Mexico}
\acmISBN{979-8-4007-0371-3/24/03}
\acmYear{2024}
\acmPrice{15.00}
\acmDOI{10.1145/3616855.3635824}

\acmSubmissionID{484}

\usepackage{amsmath,bm}  
\usepackage{caption}  
\usepackage{multirow}  
\usepackage{subcaption}  
\usepackage[normalem]{ulem}  

\newcommand{\etal}{\textit{et al}. }
\newcommand{\ie}{\textit{i}.\textit{e}., }
\newcommand{\eg}{\textit{e}.\textit{g}., }
\useunder{\uline}{\ul}{}

\begin{document}

\title{Proxy-based Item Representation for Attribute and Context-aware Recommendation}

\author{Jinseok Seol}
    \email{jamie@europa.snu.ac.kr}
    \affiliation{%
        \institution{Seoul National University}
        \city{Seoul}
        \country{Republic of Korea}
    }
\author{Minseok Gang}
    \email{alstjr3754@europa.snu.ac.kr}
    \affiliation{%
        \institution{Seoul National University}
        \city{Seoul}
        \country{Republic of Korea}
    }
\author{Sang-goo Lee}
    \email{sglee@europa.snu.ac.kr, sglee@intellisys.co.kr}
    \affiliation{%
        \institution{Seoul National University, IntelliSys Co., Ltd.}
        \city{Seoul}
        \country{Republic of Korea}
    }
\author{Jaehui Park}
    \email{jaehui@uos.ac.kr}
    \affiliation{%
        \institution{University of Seoul}
        \city{Seoul}
        \country{Republic of Korea}
    }

\renewcommand{\shortauthors}{Jinseok Seol, \etal}


\begin{abstract}

    Neural network approaches in recommender systems have shown remarkable success by representing a large set of items as a learnable vector embedding table.
    However, infrequent items may suffer from inadequate training opportunities, making it difficult to learn meaningful representations.
    We examine that in attribute and context-aware settings, the poorly learned embeddings of infrequent items impair the recommendation accuracy.
    To address such an issue, we propose a proxy-based item representation that allows each item to be expressed as a weighted sum of learnable proxy embeddings.
    Here, the proxy weight is determined by the attributes and context of each item and may incorporate bias terms in case of frequent items to further reflect collaborative signals.
    The proxy-based method calculates the item representations compositionally, ensuring each representation resides inside a well-trained simplex and, thus, acquires guaranteed quality.
    Additionally, that the proxy embeddings are shared across all items allows the infrequent items to borrow training signals of frequent items in a unified model structure and end-to-end manner.
    Our proposed method is a plug-and-play model that can replace the item encoding layer of any neural network-based recommendation model, while consistently improving the recommendation performance with much smaller parameter usage. 
    Experiments conducted on real-world recommendation benchmark datasets demonstrate that our proposed model outperforms state-of-the-art models in terms of recommendation accuracy by up to 17\% while using only 10\% of the parameters.

\end{abstract}




\maketitle

\settopmatter{printfolios=true}



\section{Introduction}

    Recommender systems have evolved from collaborative filtering (CF) \cite{sarwar2001item} to machine learning \cite{ning2011slim, rendle2010factorization, guo2017deepfm} and deep learning-based models \cite{batmaz2019review, wang2019sequential}.
    In model-based CF \cite{koren2009matrix}, users and items are represented as latent vectors via learnable embedding matrix, which is implemented as a full look-up table, where each row corresponds to a unique user or item vector.
    Such vector representation approach can transform the recommendation problem into a set of vector arithmetic, assuming that the latent vectors are well-trained through the training data distribution \cite{tran2019improving}.
    To train the representations effectively and achieve higher recommendation performance, recent studies have introduced deep learning-based models \cite{he2017neural, xue2017deep, vaswani2017attention} and sequential recommendation models \cite{quadrana2017personalizing, kang2018self, sun2019bert4rec, petrov2022systematic} as well as attribute \cite{he2016vbpr, zhou2020s3} and context-aware \cite{guo2017deepfm, rashed2022context} models.

    The concept of item embedding matrix also appears in the Natural Language Processing (NLP) field to train the distributional semantics of words \cite{mikolov2013distributed}.
    However, in contrast to the data distributions in NLP, recommendation datasets exhibit distinct and unique characteristics, which incur two significant limitations that need to be addressed.
    Firstly, the item frequency in recommendation datasets follows a long-tail distribution \cite{yin2012challenging}.
    The item embedding matrix, which is unaware of such data distribution, cannot accurately reflect the proper training signals for the infrequent items \cite{ginart2021mixed, movshovitz2017no}, leading to degraded recommendation performance \cite{ardalani2022understanding}.
    Secondly, the number of items may increase indefinitely.
    For a real-world recommendation scenario, a large set of new items is consistently added to the system  \cite{pfadler2020billion, covington2016deep}
    and the item embedding matrix, where most of the model parameters are stored \cite{ardalani2022understanding}, must scale up accordingly to accommodate them.
    For these problems being the case, it is necessary to devise a new methodology that can effectively replace the item embedding matrix and resolve such limitations.

    We further investigate the problem of infrequent item embeddings thoroughly.
    First of all, during the training process of a recommendation model, items appear in three different scenarios: a user profile that is represented as a list of interacted items, a target item for the system to predict, and a set of negative items that the user has not interacted obtained via negative sampling. 
    It is evident that the parameter update for a specific item embedding, except for that of negative sampled items, is directly proportional to the item's occurrence frequency.
    Consequently, infrequent items do not receive as many training opportunities as frequent items do and, thus, their embeddings hardly converge to their optimal states throughout the training.
    We have conducted an empirical analysis of such phenomena under attribute and context-aware settings \cite{kang2018self, rashed2022context}. 
    Specifically, we have replaced a portion of infrequent items with a shared unknown token and measured the model performance.
    In Table \ref{tab:unk}, when the infrequent items are replaced by the unknown token, which can be seen as a removal process, the recommendation accuracy in terms of NDCG@10 increases.
    Such result clearly demonstrates that the poorly learned infrequent item embeddings harm the overall performance as mentioned earlier.
    However, such improvement comes at the cost of losing global coverage of recommended items, which accounts for item diversity, and highlights that the insufficient training signals for infrequent items need to be addressed.

    Previous studies that aim to resolve the aforementioned issues can be divided into two branches.
    Firstly, to mitigate the negative impact of infrequent item embeddings on the model performance while reducing the parameter usage, \cite{chen2023clustered} employed a shared embedding through clustering and \cite{ginart2021mixed} reduced the embedding dimension of infrequent items.
    However, these methods are not suitable for datasets where content information (\ie item attributes and contexts) is crucial (\eg Fashion domain, showing large performance gap when content information is not utilized).
    Secondly, \cite{wei2021contrastive, zhu2021learning, chen2022generative} proposed methods that allow the embeddings of infrequent items to mimic those of frequent items using their content information.
    However, they have separate model structures (\eg mimicking network) and learning strategies (\eg warm-up stage), which makes it difficult to train in an end-to-end fashion, resulting in unstable hyper-parameter tuning.

    In this paper, we propose a novel \emph{proxy-based item representation} model that represents each item as a weighted sum of learnable proxy embeddings by leveraging content information, while incorporating the learning process of both frequent and infrequent items into a single framework in an end-to-end manner.
    Generally, a \emph{proxy} refers to a model or a vector learned in place of the original training objective to improve the training efficiency or performance in deep learning \cite{caron2020unsupervised, movshovitz2017no}.
    We reduce the item embedding matrix into two linearly combinable components: the proxy weighting network and the proxy embedding.
    Our model represents an item as the weighted sum of learnable proxy embeddings, where the weights are primarily determined by content information.
    As opposed to the infrequent items, frequent items possess abundant collaborative signals.
    To further incorporate such signals whenever available, we add learnable bias terms to the proxy weights.
    By such means, the model can learn hybrid collaborative signals in a single representation space that encompasses not only item ID-to-ID but also ID-to-content and content-to-content relations.

    The core concept of our proposed model is illustrated in Figure \ref{fig:model-concept}.
    We apply a softmax function on the proxy weights to enforce the weights sum up to 1.
    This approach ensures each item representation resides inside a simplex where vertices are well-trained proxy embeddings.    
    Moreover, since the gradient is formed in a unit of proxy embeddings that are shared across all items, the training becomes stable and fast while allowing infrequent items to borrow the training signal of frequent items.
    Additionally, that newly added item can be computed compositionally via proxy embeddings based on their content information prevents the indefinite increase of parameters, resulting in parameter efficiency.
    As a result, our method alleviates the inadequate parameter updates of those items that exist in long-tail distribution, leading to performance improvement with much smaller parameter usage.

    We conduct experiments on real-world recommendation benchmark datasets, namely Amazon Review datasets \cite{mcauley2015image} and MovieLens dataset.
    Each dataset has its distinct data distribution and item attributes, along with context and sequential information.
    Our proposed method is a plug-and-play model that can replace the item encoding layer of any neural network-based recommendation model, and the experimental results show that leveraging our model consistently improves the performance across all datasets.

\begin{table}[tp]
    \caption{Performance and diversity comparison of infrequent item removal ratio on Fashion}
    \label{tab:unk}
    \begin{tabular}{cc|c|cc}
        \toprule
            Model & Removal & \#params (M) & NDCG@10 & Diversity \\
        \midrule
                & 0\%  & 45.8 & 37.9\% & 15.0\% \\
                & 25\% & 35.1 & 39.1\% & 14.9\% \\
            SASRec++
                & 50\% & 24.5 & 39.2\% & 14.8\% \\
                & 75\% & 13.9 & 40.2\% & 13.1\% \\
                & 90\% &  7.5 & 41.6\% & 11.6\% \\
        \midrule
                & 0\%  & 25.1 & 42.6\% & 16.2\% \\
                & 25\% & 19.7 & 42.5\% & 15.3\% \\
            with \textbf{PIR}
                & 50\% & 14.4 & 42.5\% & 16.2\% \\
                & 75\% &  9.1 & 42.4\% & 16.2\% \\
                & 90\% &  5.9 & 42.4\% & 16.0\% \\ 
        \bottomrule
    \end{tabular}
\end{table}

\begin{figure*}[tp]
    \centering
    \includegraphics[width=0.8\linewidth]{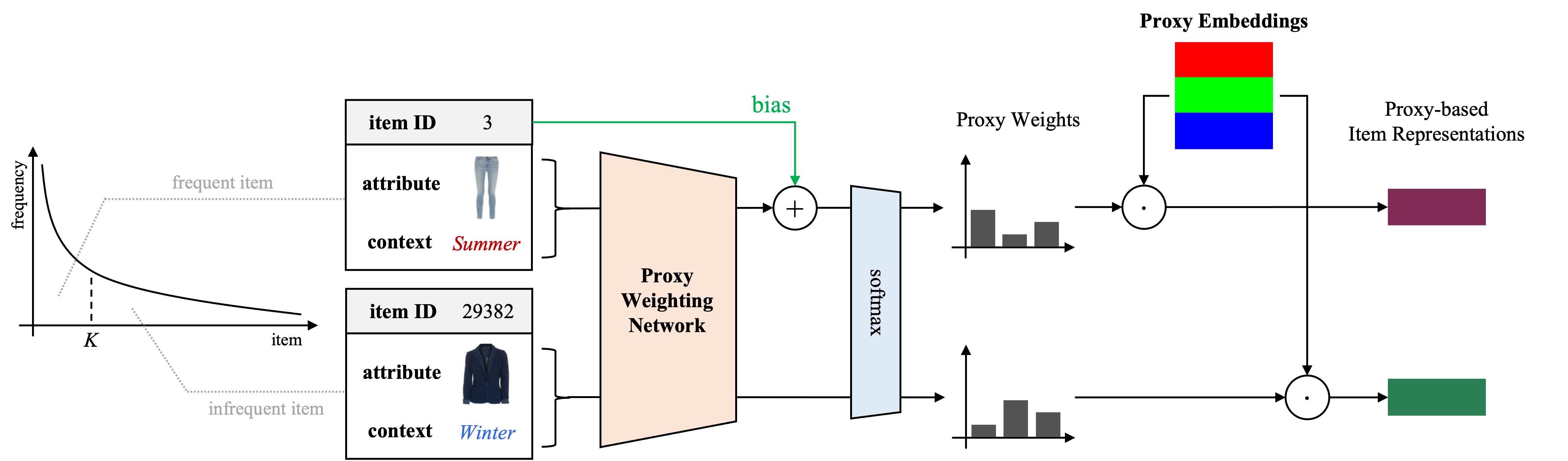}
    \caption{Proxy-based item representation model. Items are represented as a weighted sum of learnable proxy embeddings.}
    \label{fig:model-concept}
\end{figure*}


\section{Related Work}

    Our model is inspired by three research branches: attribute and context-aware sequential recommendation (ACSR) models \cite{zhou2020s3, petrov2022effective, petrov2022systematic, chen2022intent}, parameter-efficient recommendation models \cite{ginart2021mixed, kang2021learning, liu2021learnable}, and proxy learning models \cite{kim2020proxy, cho2021unsupervised} involving clustering techniques in self-supervised learning \cite{caron2020unsupervised, li2020prototypical}.

    The most important feature in the recommendation is the user-item interaction history.
    However, in certain datasets where interaction data is sparse or item attributes are crucial, leveraging content and contextual information becomes vital \cite{guo2017deepfm, zhou2020s3, wei2021contrastive, seol2022exploiting}.
    For example, the current state-of-the-art ACSR model CARCA \cite{rashed2022context} achieves a performance increase of up to 50\% in terms of recommendation accuracy when compared to other baseline models that do not reflect all of the following: the attribute, the context, or the sequential information.
    Furthermore, in the cold-start scenarios where the interaction data are extremely rare, \cite{wei2021contrastive, zhu2021learning, chen2022generative} proposed two-structured methods that aim to mimic the embeddings of frequent items by employing content information to boost model performance.

    To address the enormous training parameter issue \cite{ardalani2022understanding}, studies for parameter-efficient recommendation models have been conducted recently \cite{kang2021learning, ginart2021mixed, ardalani2022understanding, liu2021learnable, yan2021learning, zhao2020memory}.
    Most studies focus on devising new methodologies to encode each item into a promising vector representation with fewer parameters \cite{kang2021learning, ginart2021mixed, zhao2020memory, chen2023clustered} or reduce the model size without hurting the performance \cite{shen2021umec, yan2021learning}.
    In the mixed-dimension model \cite{ginart2021mixed}, the item embedding matrix is split into two parts: frequent and infrequent items, where the embeddings of the latter are factorized into low-rank matrices.
    Such method has the potential to alleviate the negative impact of infrequent item embeddings on the model performance while reducing the overall parameter usage.
    In the hash-based method \cite{kang2021learning}, the embedding matrix is compressed through multiple hash functions, representing the item embeddings as the combination of hash values.
    The clustering-based method of \cite{chen2023clustered} also aims to resolve inherent issues of the embedding matrix by using the shared embeddings.
    However, such methods does not reflect essential recommendation properties, such as the attributes of items and the recommendation contexts, limiting performance improvement in certain datasets that heavily resort to item content information (\eg Fashion domain).

    The concept of utilizing proxies has been applied and verified in many studies from various fields with similar motivations: infrequently occurring objects borrowing the training signals from frequently occurring objects.
    Among numerous recommendation models, ProxySR \cite{cho2021unsupervised} employs user proxy embeddings to augment information for item sequences with a relatively short length in a session-based setting.
    However, ProxySR explicitly selects one proxy user at a time, which is very different from our proposed model that represents items by combining multiple proxy item embeddings.
    Concretely, proxies in ProxySR represent prototypical users, while proxies in our model represents cluster centroids of items that serve as necessary components for item representations.
    Despite such differences, ProxySR successfully demonstrates that the proxy embeddings serve as an external memory \cite{dong2020mamo} of collaborative signals \cite{li2021lightweight, fan2021continuous} and provide a performance improvement.

    In computer vision, proxies are also utilized as cluster centroids to improve training efficiency and performance in contrastive learning, namely self-supervised learning \cite{caron2020unsupervised, caron2021emerging, li2020prototypical, zhang2021supporting, cho2021masked} and metric learning for image retrieval \cite{wieczorek2021unreasonable, movshovitz2017no, kim2020proxy, yao2022pcl}.
    Contrastive learning involves positive and negative sampling, where the complexity of image pairs or triplets becomes enormous \cite{kim2020proxy}.
    This leads to inconsistent parameter updates since some positive samples may starve by chance.
    To address such problem, training is performed between proxies rather than image samples.
    In this setting, a proxy represents a centroid of an unsupervised cluster or a representative vector of a supervised labeled group.
    The individual images are then trained via their corresponding proxy.
    Since the number of proxies is much smaller than the number of images, the network converges faster, and the training becomes more stable while improving the overall performance \cite{kim2020proxy}.
    Regarding the recommendation, the infrequent items often undergo insufficient parameter updates, corresponding to the aforementioned inconsistent parameter update of positive samples in the contrastive learning.
    Therefore, employing proxies that behave similarly to cluster centroids can alleviate such training opportunity starvation of infrequent items.


\section{Model}

    \subsection{Problem Definition}
    
        Let $U$ and $I$ be the set of all users and all items, respectively.
        The \emph{user item sequence} $S_u$ of a user $u \in U$ is defined as $S_{u} = (i_{1}, i_{2}, ..., i_{|S_{u}|})$ where $1 \leq t \leq |S_{u}|$ and $i_{t} \in I$.
        We denote a subsequence that uses only the first $t'$ items as $S_{u, (t \leq t')}$.
        We define negative items as $I_{u}^{-} = I \setminus I_{u}^{+}$ where $I_{u}^{+}$ is a set of items constituting $S_{u}$.
        The goal of \emph{sequential recommendation} (SR) is to predict subsequent items with a given user item sequence \cite{kang2018self}.
        For every $u \in U$ and each $1 \leq t' \leq |S_{u}| - 1$, the objective is to maximize $P(i_{t' + 1} | S_{u, (t \leq t')})$, when compared to negative items $N_{u} = \{i_{j}^{-}\}_{j=1}^{|N_{u}|}$ that are randomly sampled from $I_{u}^{-}$.
        We assume that each item has auxiliary information in addition to its interaction history, which is referred to as \emph{item attributes}.
        The item attributes exist in various forms, such as images, titles, descriptions, keyword tags, etc.
        Here, we assume that preprocessing for attributes is done, and we can use it as a $d_{\text{A}}$-dimensional vector $f_{i_{t}}$ for the item $i_{t}$.
        Meanwhile, a \emph{context} is the contextual information, generally meaning time, place, situation, etc., that can affect the preference independently of users and items.
        Similar to attributes, we assume that the context is preprocessed as a $d_{\text{C}}$-dimensional vector $c_{t}$.
        We extend SR into \emph{attribute and context-aware sequential recommendation} (ACSR) as follows: maximize $P(i_{t' + 1} | S_{u, (t \leq t')}^{\text{AC}}, c_{t' + 1})$ where $S_{u}^{\text{AC}} = \{(i_{t}, f_{i_{t}}, c_{t}) | 1 \leq t \leq |S_{u}^{\text{AC}}|\}$.

    \subsection{Background}

        The general architecture of an ACSR model can be divided into three parts: the item encoding layer, the sequence encoding blocks, and the item scoring layer.
        We explain each part based on the design of CARCA \cite{rashed2022context}, where it has (1) a strong item encoding layer that utilizes item attributes and contexts, (2) self-attention-based sequence encoding blocks, and (3) a state-of-the-art item scoring layer using a cross-attention.

        \subsubsection{Item Encoding Layer}

            The main function of this layer is to extract the \emph{item vector} $v_{i_{t}} \in \mathbb{R}^{d}$ given the input $(i_{t}, f_{i_{t}}, c_{t})$, where $d$ denotes the latent dimension.
            First, the individual \emph{item embedding} for the item $i_{t}$ is assigned as $\text{IE}_{i_{t}}$, from the full-item look-up table $\text{IE} \in \mathbb{R}^{|I| \times d_{\text{IE}}}$, where $d_{\text{IE}}$ denotes the dimension of the item embedding.
            These three vectors $(\text{IE}_{i_{t}}, f_{i_{t}}, c_{t})$ are then passed to a shallow neural network to obtain the final item vector:
            \begin{equation}
                 z_{i_{t}} = \sigma^{\text{AC}}(\text{cat}_{\text{col}}(f_{i_{t}}, c_{t})W^{\text{AC}} + b^{\text{AC}}),
            \end{equation}
            \begin{equation}
                 v_{i_{t}} = \sigma^{\text{item}}(\text{cat}_{\text{col}}(\text{IE}_{i_{t}}, z_{i_{t}})W^{\text{item}} + b^{\text{item}}),
            \end{equation}
            where $W^{\text{AC}} \in \mathbb{R}^{(d_{A} + d_{\text{C}}) \times d_{\text{AC}}}$, $b^{\text{AC}} \in \mathbb{R}^{d_{\text{AC}}}$, $W^{\text{item}} \in \mathbb{R}^{(d_{\text{IE}} + d_{\text{AC}}) \times d}$ , $b^{\text{item}} \in \mathbb{R}^{d}$ are weights and biases for corresponding network, $\sigma^{\text{AC}}$, $\sigma^{\text{item}}$ are activation functions, and $\text{cat}_{\text{col}}$ denotes the column-wise vector concatenation.

        \subsubsection{Sequence Encoding Blocks}

            With the advent of self-attention blocks, the sequence encoding blocks have been commonly implemented with Transformer \cite{vaswani2017attention} architecture that is superior in terms of both performance and computational efficiency to that of the past, namely RNN or CNN-based methods \cite{quadrana2017personalizing, kang2018self}.
            An \emph{attention} (Attn) is defined as follows: for three given sequences of vectors $Q \in \mathbb{R}^{L_{Q} \times d_{Q}}$, $K \in \mathbb{R}^{L_{K} \times d_{K}}$, $V \in \mathbb{R}^{L_{V} \times d_{V}}$,
            \begin{equation}
                 \text{Attn}(Q, K, V) = \text{softmax} \left( (Q K^{\top})/\sqrt{d_{QK}} \right) V,
            \end{equation}
            where $d_{Q} = d_{K} = d_{QK}$ and $L_{K} = L_{V} = L_{KV}$.
            Given the number of heads $H$ that is the divisor of $d_{Q}$, $d_{K}$, $d_{V}$, we can create separate linear projection layers so that each head can handle different representations, often referred \emph{multi-head attention} (MHA):
            \begin{equation}
                 \text{MHA}(Q, K, V) = \text{cat}_{\text{col}} \left( \left [ \text{Attn} (Q W_{h}^{Q}, K W_{h} ^{K}, V W_{h}^{V}) \right ]_{h=1}^{H} \right),
            \end{equation}
            where $W_{h}^{Q} \in \mathbb{R}^{d_{Q} \times d_{Q}/H}$, $W_{h}^{K} \in \mathbb{R}^{d_{K} \times d_{K}/H}$, $W_{h}^{V} \in \mathbb{R}^{d_{V} \times d_{V}/H}$.
            The \emph{self-attention} assumes that $Q$, $K$, and $V$ are the same, and it is used to encode the complex information of the input sequence.
            In our problem setting, we apply self-attention to item vectors $V_{u} = [v_{i_{1}}^{\top}, v_{i_{2}}^{\top}, ..., v_{i_{t'}}^{\top}]^{\top} \in \mathbb{R}^{|S_{u, (t \leq t')}^{\text{AC}}| \times d}$:
            \begin{equation}
                 \hat{V}_{u} = \text{MHA}(V_{u}, V_{u}, V_{u}).
            \end{equation}
            Following the architecture of Transformer, the results above are passed to a two-layered point-wise feed-forward network (PWFF), where it is defined as follows: for a matrix $X \in \mathbb{R}^{L_{X} \times d_{X}}$ denoting a set of vectors,
            \begin{equation}
                 \text{PWFF}(X) = \text{cat}_{\text{row}}\left(\left[ \sigma(X_{r} W^{(1)} + b^{(1)})W ^{(2)} + b^{(2)} \right]_{r=1}^{L_{X}}\right),
            \end{equation}
            where $\sigma$ is an activation function, $W^{(1)}, W^{(2)} \in \mathbb{R}^{d_{X} \times d_{X}}$, $b^{(1)}, b^{(2)} \in \mathbb{R}^{d_{X}}$ are weights and biases, and $\text{cat}_{\text{row}}$ is a row-wise vector concatenation.
            PWFF is an additional layer that helps the model understand more complex relationships and provides further non-linearity.
            Assuming an additive residual connection, the self-attention block can be stacked into multiple blocks as follows:
            \begin{equation}
                 \hat{V}_{u}^{(b)} = \text{MHA}(V_{u}^{(b)}, V_{u}^{(b)}, V_{u}^{ (b)}),
            \end{equation}
            \begin{equation}
                 V_{u}^{(b + 1)} = \text{PWFF}(\hat{V}_{u}^{(b)}) + \hat{V}_{u}^{(b) },
            \end{equation}
            where $V_{u}^{(1)} = V_{u}$, and $(b)$ denoting the $b$-th block, up to the total number of $B$ attention blocks.
            The final vectors $V_{u}^{(B)}$ from the sequence encoding blocks are dubbed \emph{latent sequence vectors}.
            Unlike many other attention-based models, since the temporal information is often explicitly given in the form of a context \cite{seol2022exploiting, rashed2022context}, we do not employ positional embeddings.

        \subsubsection{Item Scoring Layer}

            After extracting the latent sequence vectors, the \emph{item scoring layer} calculates preference scores to rank candidate items, producing recommendation output.
            There are two major methodologies for the layer.
            
            \paragraph{Inner Product (IP)}
                Similar to MF, this method uses the inner product (IP) value between a user vector and an item vector to obtain a preference score.
                Since the IP method is commonly implemented via negative item sampling, it is widely used due to its great advantage in computational efficiency, especially in terms of memory and parameters \cite{mikolov2013distributed, kang2018self}.
                The sampled candidate items are encoded into \emph{candidate item vectors} $C_{u} = [v_{i_{t'+1}}^{\top}, v_{i_{ 1}^{-}}^{\top}, v_{i_{2}^{-}}^{\top}, ..., v_{i_{|N_{u}|}^{-}}^{\top}]^{\top} \in \mathbb{R}^{(|N_{u}| + 1) \times d}$ via the same item encoding layer used in the user-item sequence encoding, with shared parameters.
                The training method of IP can be categorized into two losses \cite{chen2020simple}:
                \begin{enumerate}
                    \item Binary cross-entropy loss (BCE): This method trains positive and negative items independently by treating the IP value as the logit of binary classification.
                    \item Normalized temperature-scaled cross-entropy loss (NT-Xent) \cite{chen2020simple, sohn2016improved}: After the $L_{2}$-normalization to make IP as cosine similarity, a temperature-scaled softmax is applied to perform $(N+1)$-way classification. This method trains the positive score to be relatively higher than the negative score.
                \end{enumerate}
            
            \paragraph{Cross-Attention (CA)}
                Unlike the IP, the sequence latent vector is not treated as the user vector at step $t$ in CA \cite{rashed2022context}.
                The method uses candidate item vectors $C_{u}$ as query and sequence latent vectors $V_{u}^{(B)}$ as both key and value for another multi-head attention layer.
                After applying a multiplicative residual connection, the results are passed to a scoring layer $\phi^{\text{score}}$ that produces the preference scores:
                \begin{equation}
                     \hat{C}_{u} = C_{u} \odot \text{MHA} (C_{u}, V_{u}^{(B)}, V_{u}^{(B)}),
                \end{equation}
                \begin{equation}
                     Y = \phi^{\text{score}}(\hat{C}_{u}) = \sigma^{\text{score}} ( \hat{C}_{u} W^{\text{score}} + b^{\text{score}} ),
                \end{equation}
                where $W^{\text{score}} \in \mathbb{R}^{d}$ is a scoring weight vector, $b^{\text{score}} \in \mathbb{R}$ is a scoring bias, and $\sigma^{\text{score}}$ is the activation function.
                The score is treated as the logit of binary classification, similar to BCE from the IP method.
                When CA decides whether or not to recommend each candidate item, it explicitly considers the entire item sequence via the cross-attention mechanism, leading to superior performance in most cases.

    \subsection{Proxy-based Item Representation}

        Our proposed model, the \textbf{P}roxy-based \textbf{I}tem \textbf{R}epresentation (PIR) method, replaces the item embeddings $\text{IE}_{i_{t}}$ in the item encoding layer.
        First, we introduce the learnable \emph{proxy embeddings}, namely $P = [p_{1}^{\top}, p_{2}^{\top}, ..., p_{n_{\text{proxy}}}^{\top}]^{\top} \in \mathbb{R}^{n_{\text{proxy}} \times d_{\text{proxy}}}$.
        Here, the number of proxy embeddings $n_{\text{proxy}}$ and its dimension $d_{\text{proxy}}$ are hyper-parameters.
        We calculate the appropriate weights for each proxy embedding by item attribute $f_{i_{t}}$ and context $c_{t}$ to produce a \emph{proxy-based item representation} $\text{PIR}_{i_{t}} \in \mathbb{R}^{d_{\text{proxy}}}$ that is a weighted sum of proxy embeddings.
        Specifically, we use a 2-layered neural network $\varphi$ to compute the unnormalized weight for each proxy embedding:
        \begin{equation}
            w_{i_{t}}' = \sigma^{\varphi, (1)} ( \text{cat}_{\text{col}}(f_{i_{t}}, c_{t} ) W^{\varphi, (1)} + b^{\varphi, (1)} ),
        \end{equation}
        \begin{equation}
            \begin{split}
                w_{i_{t}} &= \sigma^{\varphi, (2)} ( w_{i_{t}}' W^{\varphi, (2)} + b^{\varphi, (2)} ) \\
                &= \varphi (f_{i_{t}}, c_{t}) \in \mathbb{R}^{n_{\text{proxy}}},
            \end{split}
        \end{equation}
        where $W^{\varphi, (1)} \in \mathbb{R}^{(d_{\text{A}} + d_{\text{C}}) \times d_{\varphi}}$, $b^{\varphi, (1)} \in \mathbb{R}^{d_{\varphi}}$ are the weight and bias for the first layer, and $W^{\varphi, (2)} \in \mathbb{R}^{d_{\varphi} \times n_{\text{proxy}}}$, $b^{\varphi , (2)} \in \mathbb{R}^{n_{\text{proxy}}}$ are the weight and bias for the second layer.
        $\sigma^{\varphi, (1)}$ and $\sigma^{\varphi, (2)}$ can be any appropriate activation functions, but we use LeakyReLU for $\sigma^{\varphi, (1)}$ and the identity function for $\sigma^{\varphi, (2)}$, which were empirically chosen through hyper-parameter tuning.
        After normalizing through a softmax layer, $\text{PIR}_{i_{t}}$ is computed as the weighted sum of proxy embeddings:
        \begin{equation}
             \hat{w}_{i_{t}} = \text{softmax}(w_{i_{t}}),
        \end{equation}
        \begin{equation}
             \text{PIR}_{i_{t}} = \hat{w}_{i_{t}} P.
        \end{equation}
        The final item vector $v_{i_{t}}$ is then calculated by appending the item attribute and context information to the newly created representation.
        Slightly different from the process described in CARCA, we introduce an additional 1-layer neural network for attribute vectors to provide an additional non-linearity to the attribute information.
        \begin{equation}
             f_{i_{t}}' = \phi^{\text{A}}(f_{i_{t}}) = \sigma^{\text{A}}(f_{i_{t}} W^{\text{A}} + b^{\text{A}}),
        \end{equation}
        \begin{equation}
             z_{i_{t}} = \phi^{\text{AC}}(f_{i_{t}}', c_{t}) = \sigma^{\text{AC}}(\text{cat}_{\text{col}}(f_{i_{t}}', c_{t})W^{ \text{AC}} + b^{\text{AC}}),
        \end{equation}
        \begin{equation}
             v_{i_{t}} = \phi^{\text{item}}(i_{t}, z_{i_{t}}) = \sigma^{\text{item}}(\text{cat}_{\text{col}}(\text{PIR}_{i_{t}}, z_{i_{t}})W ^{\text{item}} + b^{\text{item}}),
        \end{equation}
        where $W^{\text{A}} \in \mathbb{R}^{d_{\text{A}} \times d_{\text{A}}'}$, $b^{\text{A}} \in \mathbb{R}^{d_{\text{A}}'}$, $W^{\text{AC}} \in \mathbb{R}^{(d_{\text{A}}' + d_{\text{C}}) \times d_{\text{AC}}}$, $b^{\text{AC}} \in \mathbb{R}^{d_{\text{AC}}}$, $W^{\text{item}} \in \mathbb{R}^{(d_{\text{proxy}} + d_{\text{AC}}) \times d}$, $b^{\text{item}} \in \mathbb{R}^{d}$ are the weights and biases for the corresponding networks and $\sigma^{\text{A}}$, $\sigma^{\text{AC}}$, $\sigma^{\text{item}}$ are the activation functions.
        The overall architecture of our model is illustrated in Figure \ref{fig:model-architecture}.

        The above structure itself represents items solely based on their content information and behaves similarly to content-based filtering, which cannot reflect the collaborative signals between the items on its own.
        To give direct collaborative signals to frequent items, we introduce a \emph{frequent item bias}, a concept similar to the known user bias from \cite{cho2021unsupervised}.
        The frequent item bias is a structure where the selected top $K$ frequent items can memorize the biases for the proxy weights:
        \begin{equation}
            \hat{w}_{i_{t}} = \text{softmax}(w_{i_{t}} + b_{i_{t}}^{\text{freq}}),
        \end{equation}
        where $b_{i_{t}}^{\text{freq}} \in \mathbb{R}^{n_{\text{proxy}}}$ is a learnable bias vector.
        Existing models often attempt to reduce parameter usage by allocating different dimensions \cite{ginart2021mixed, zhao2020memory} or model structures \cite{wei2021contrastive, zhu2021learning} for frequent and infrequent items.
        In our proposed model, the only difference between the frequent and infrequent items is whether it can partly memorize the proxy weights or not in an identical model structure.
        Note that if $K = 0$, the proxy-based model can be interpreted as a deep learning version of content-based filtering, since it only uses the content information.
        On the other hand, if $K = |I|$, the model behaves similarly to low-rank factorization of the full-item look-up table $\text{IE}$, with a latent dimension of $n_{\text{proxy}}$, since the computation is similar to matrix factorization but with softmax non-linearity.
        Note that in this case, the model can still be applied to the case where there is no content information available, which reduces exactly into the low-rank factorization with softmax non-linearity.

        Our choice of using the softmax function generates $\text{PIR}_{i_{t}}$ to be a vector that resides inside a simplex, where vertices are proxy embeddings: $\hat{w}_{i_{t}} P = \sum_{r = 1}^{n_{\text{proxy}}} \hat{w}_{i_{t}, r} P_{r}$ where $\sum_{r = 1}^{n_{\text{proxy}}} \hat{w}_{i_{t}, r} = 1$, meeting the condition for simplex.
        Here, we can expect that the abundant interaction of frequent items will train the proxy embeddings to their promising state.
        Therefore, any $\text{PIR}_{i_{t}}$ can be expected to be a comprehensible vector to the network, where otherwise the infrequent item representations would have been learned poorly due to inadequate training.
        In addition, another property of the weighted sum mechanism is that the gradient of parameter update is formed as a unit of proxy embeddings: assuming that the proxy embeddings are fixed, $\frac{\partial \text{PIR}_{i_{t}}}{\partial \theta} = \sum_{r = 1}^{n_{\text{proxy}}} \frac{\partial \hat{w}_{i_{t}, r}}{\partial \theta} P_{r}$, for a parameter $\theta$ to be updated.
        This property allows infrequent items to borrow training signals from frequent items via proxy embeddings.
        We also prove the following propositions in the appendix: \emph{content locality}, where infrequent items with similar attributes are close to each other in the proxy-based representation space, and \emph{bias priority}, where the frequent item bias can be prioritized over the content information.

\begin{figure}[tp]
    \centering
    \includegraphics[width=0.95\linewidth]{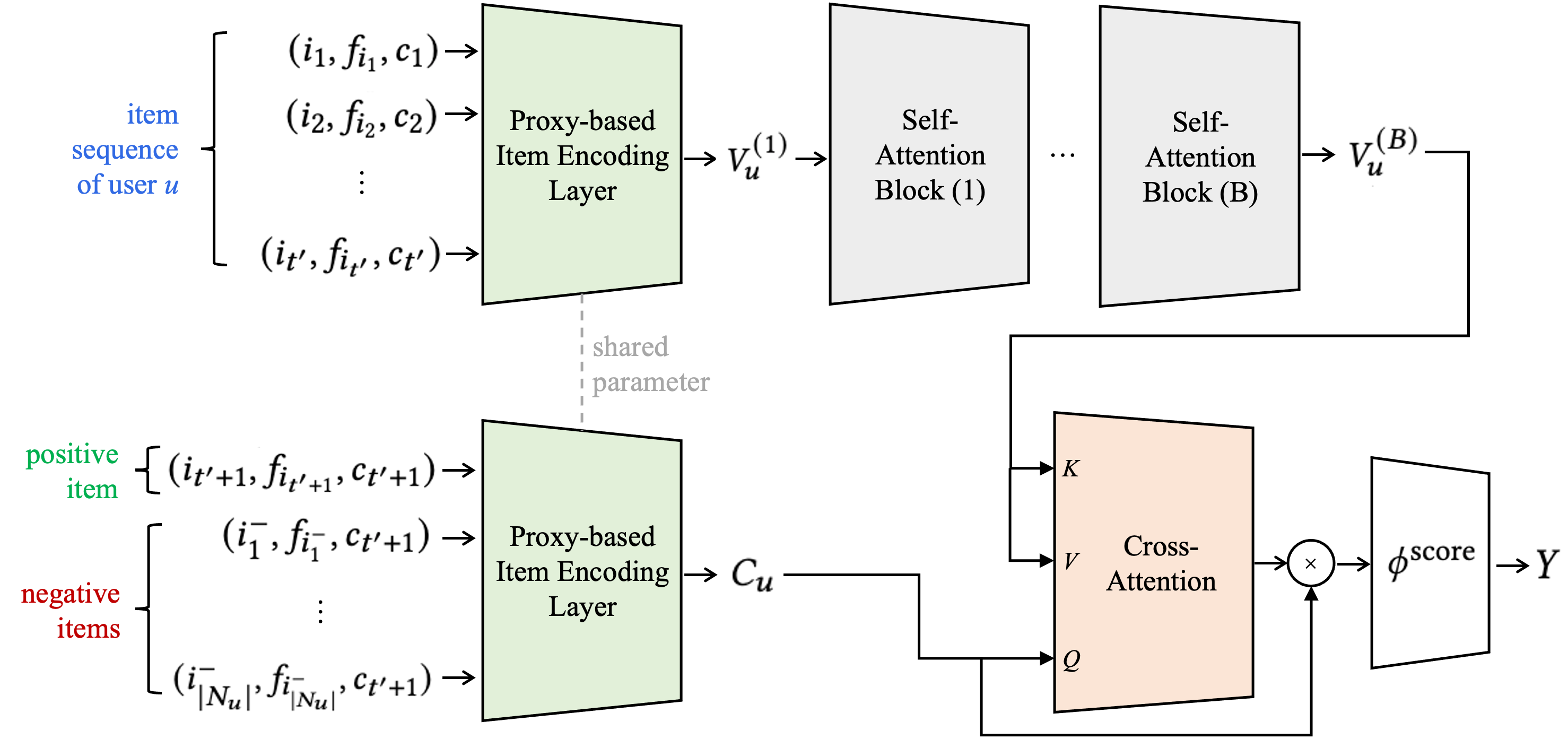}
    \caption{Overall model architecture.}
    \label{fig:model-architecture}
\end{figure}

\begin{table*}[tp]
    \caption{Dataset statistics after preprocessing.}
    \label{tab:dataset}
    \begin{tabular}{c|rrrrrrr}
        \toprule
            Dataset & \#user & \#item & \#interaction & Density & Unique Attributes & Duplicates & Memorization \\
        \midrule
            Fashion &  45,184 & 166,270 &   358,003 & 0.0048\% & 163,985 &  1.0 & 98.09\% \\
            Men     &  34,244 & 110,636 &   254,870 & 0.0067\% & 109,282 &  1.0 & 98.36\% \\
            Beauty  &  52,204 &  57,289 &   394,908 & 0.0132\% &  11,094 &  5.2 & 41.92\% \\
            Game    &  31,013 &  23,715 &   287,107 & 0.0390\% &   1,701 & 13.9 &  7.26\% \\
            ML-20M  & 138,287 &  20,720 & 9,995,410 & 0.3488\% &   1,342 & 20.3 &  4.92\% \\
        \bottomrule
    \end{tabular}
\end{table*}

    \subsection{Complexity of Parameters}

        Since $\text{IE}$ uses $|I|$ embeddings while $P$ only uses $n_\text{proxy}$ embeddings, the proxy-based item representation can reduce a significant number of parameters.
        For computational convenience, assume that all latent dimensions are equal to, or proportional to $d$: $d_{\text{IE}} = d$, $d_{\text{AC}} \propto d$, $d_{\text{A}}' = d$, $d_{\varphi} \propto d$, and $d_{\text{proxy}} = d$.
        In a baseline model that uses $\text{IE}$, the number of parameters is $O(|I|d + M(d))$, where $M(d)$ is the number of parameters of the network other than the item embedding.
        For our model, the number of parameters is $O((K + d)n_{\text{proxy}} + M'(d))$, where $M'(d)$ denotes the number of parameters of the network other than the PIR layer.
        Note that in this case, the number of parameters does not increase proportionally to the number of items $|I|$.
        Even in the extreme case where $K = |I|$, parameters can be reduced whenever $n_{\text{proxy}}$ is just slightly smaller than $d$: let $n_{\text{proxy}} = \alpha d$, then we get $|I|d + M(d) > (|I| + d)n_{\text{proxy}} + M'(d) \Leftrightarrow (|I| + d)^{-1}\left(|I| + (M(d) - M'(d)) / d \right) > \alpha$.
        Assuming that $M(d) \approx M'(d)$, the model can shrink whenever $\alpha$ meets the condition (note that $d \ll |I|$ holds in general).

    \subsection{Training Objective}
    
        In sequential recommendation, it is common to create $|S_{u}^{\text{AC}}|-1$ tasks of Next Item Prediction (NIP) by using the aforementioned latent sequence vectors as user vectors, namely $V_{u, t}^{(B)}$ for each time step $1 \leq t \leq |S_{u}^{\text{AC}}|-1$.
        However, in the case of cross-attention, the latent sequence vectors possess bidirectional information of the item sequence.
        Therefore, training cross-attention with NIP task incurs a discrepancy between the training phase and the inference phase since NIP assumes a causality constraint in the training phase.
        To this end, we adopt Last Item Prediction (LIP) task, which uses only the very last item of each user as a positive item.
        In every training epoch, each item sequence $S_{u}^{\text{AC}}$ is randomly cut \cite{wu2020sse} into a subsequence so that diverse input data can be trained for the same user.
        We also apply NT-Xent \cite{rendle2009bpr, sohn2016improved} on the cross-attention, so that the model can employ multiple negatives.
        We call this method Contrastive CA, where it can only be implemented in LIP task since it is not feasible to sample multiple negative items having the same context in NIP task due to the enormous memory cost.
        The final training objective is as follows: Assuming that for every epoch, each item sequence $S_{u}^{\text{AC}}$ of user $u$ is randomly cut, and $N_{u}$ is randomly sampled except for testing,
        \begin{equation}
             \mathcal{L} = -|U|^{-1} \textstyle \sum_{u \in U} \log \left ( e^{Y_{1}} / \textstyle \sum_{j = 1}^{|I_{u }^{-}| + 1} e^{Y_{j}} \right).
        \end{equation}


\section{Experiment}

    To demonstrate the effectiveness of our proposed model, we have constructed experiments to answer the following research questions:
    \begin{itemize}
        \item \textbf{RQ1} Does PIR layer show superior performance in a plug-and-play manner?
        \item \textbf{RQ2} Does PIR layer enhance representation quality of infrequent items?
        \item \textbf{RQ3} Does ProxyRCA have a parameter-efficient structure?
        \item \textbf{RQ4} Do the proxy weights and proxy embeddings together compositionally compute item representations?
    \end{itemize}

    \subsection{Experimental Settings}

        \subsubsection{Datasets}
    
            To evaluate and analyze the proposed model, we use five widely used recommendation datasets from different domains, namely Fashion, Men, Game, and Beauty from Amazon Review datasets \cite{mcauley2015image}, and MovieLens-20M (ML-20M).
            Fashion and Men use the image vectors extracted from the pre-trained ResNet50 model as their item attributes.
            Except for the price feature of Game, the attributes for Beauty and Game are discrete and categorical information, namely tags.
            In ML-20M, only the genre of each movie is used as item attributes, which represents the case where the content information are extremely limited.
            For contexts, timestamp data is decomposed into multi-dimensional date information, following the preprocessing steps in \cite{rashed2022context}.
            The overall statistics of the datasets are given in Table \ref{tab:dataset}.
            In the preprocessing step, users with three or fewer interactions are excluded from the dataset to make a proper train-valid-test split.

\begin{table*}[tp]
    \caption{Performance comparison on all datasets. The best and the second best results are marked as bold and italic numbers respectively. The asterisk(*) denotes statistically significant (p < 0.05) gain against the non-PIR counterpart, using the $t$-test.}
    \label{tab:main}
    \begin{tabular}{cl|cccccccccc}
        \toprule
            \multirow{2}{*}{Task} & \multirow{2}{*}{Model} & \multicolumn{2}{c}{Fashion} & \multicolumn{2}{c}{Men} & \multicolumn{2}{c}{Beauty} & \multicolumn{2}{c}{Game} & \multicolumn{2}{c}{ML-20M} \\
            & & R@10 & N@10 & R@10 & N@10 & R@10 & N@10 & R@10 & N@10 & R@10 & N@10 \\
        \midrule
            & Popular    & 0.407 & 0.262 & 0.415 & 0.269 & 0.451 & 0.261 & 0.519 & 0.314 & 0.815 & 0.530 \\
        \midrule
            \multirow{2}{*}{\begin{tabular}{c}BPR\end{tabular}}
            & BPR++
                & 0.523 & 0.332 & 0.429 & 0.266 & 0.505 & 0.352 & 0.768 & 0.564 & 0.956 & 0.702 \\
            & $\, \hookrightarrow$ with \textbf{PIR}
                & 0.620* & 0.406* & 0.554* & 0.358* & 0.511 & 0.353 & 0.770 & {\ul 0.582}* & 0.957 & 0.761* \\
        \midrule
            \multirow{5}{*}{\begin{tabular}{c}LIP\end{tabular}}
            & MixDim++
                & 0.623 & 0.407 & 0.570 & 0.365 & 0.587 & 0.398 & 0.766 & 0.556 & {\ul 0.961} & 0.781 \\
            & SASRec++
                & 0.630 & 0.416 & 0.587 & 0.379 & 0.601 & 0.415 & 0.753 & 0.539 & 0.949 & 0.761 \\
            & $\, \hookrightarrow$ with \textbf{PIR}
                & 0.635 & 0.426* & 0.580 & 0.381 & 0.599 & 0.422 & {\ul 0.779}* & 0.572* & 0.944 & 0.773* \\
            & CARCA
                & {\ul 0.648} & {\ul 0.427} & {\ul 0.614} & {\ul 0.398} & {\ul 0.608} & {\ul 0.423} & 0.762 & 0.560 & 0.961 & {\ul 0.788} \\
            & $\, \hookrightarrow$ with \textbf{PIR} (\textbf{ProxyRCA})
                & \textbf{0.661}* & \textbf{0.446}* & \textbf{0.617} & \textbf{0.408} & \textbf{0.626}* & \textbf{0.449}* & \textbf{0.809}* & \textbf{0.611}* & \textbf{0.962} & \textbf{0.792}* \\
        \bottomrule
    \end{tabular}
\end{table*}

        \subsubsection{Evaluation Protocol}

            Overall, we follow the conventional evaluation protocol from the sequential recommendation studies \cite{rashed2022context, he2016vbpr, kang2018self}.
            We use the leave-one-out evaluation where the last item of each user is the test item, and the previous one to the last item is the valid item.
            A total of 100 negative items are randomly sampled among items that have not been interacted with the user.
            The performance is measured with Hit-Ratio (HR) and Normalized Discounted Cumulative Gain (NDCG), averaged over all users.
            We run 5 times each with different random seeds and report the average performance.

        \subsubsection{Comparison Models}

            For the fair comparison, we implemented the baseline models to have a similar structure and the same training objective (LIP task) to our setting.
            Full explanation and experimental results on NIP task baselines, namely BERT4Rec \cite{sun2019bert4rec}, SASRec \cite{kang2018self}, SSE-PT \cite{wu2020sse}, and S$^3$Rec \cite{zhou2020s3} are in the appendix.

            \begin{enumerate}
                \item Popular: A non-personalized recommendation where the preference score is based on the item's global popularity. It serves as a sanity check for the performance lower bound.
                \item SASRec++ \cite{rashed2022context}: An extension of SASRec \cite{kang2018self} that utilizes item attributes and context.
                \item CARCA \cite{rashed2022context}: A state-of-the-art ACSR model that uses cross-attention as an item scoring layer.
                \item BPR++ (ours): An extension of the non-sequential model BPR \cite{rendle2009bpr} that utilizes item attributes and context.
                \item MixDim++ (ours): An extension of CARCA with parameter-efficient mixed-dimension embedding \cite{ginart2021mixed} for encoding infrequent items. Note that this is an important baseline, which represents models that handle infrequent items (either long-tail or cold-start).
                \item ProxyRCA (ours): This is our main proposed model that employs proxy-based item representation as the item encoding layer on top of CARCA, named after \textbf{Proxy}-based item representation \textbf{R}ecommendation model with \textbf{C}ross-\textbf{A}ttention.
            \end{enumerate}

        \subsubsection{Implementation Details}

            To match the similar scale of parameters, we fixed the latent dimension $d$ to 256, and $n_{\text{proxy}}$ was set to 128.
            We initialized $b_{i_{t}}^{\text{freq}}$ to be a zero vector in order to prevent the unintended influence on the computation of proxy weights in the early stages of the training.
            Detailed hyper-parameter tuning procedures are described in the appendix.
            The PyTorch-based implementation is shared as an open-source repository \footnote{https://github.com/theeluwin/ProxyRCA}.

    \subsection{Overall Performance Comparison (RQ1)}

        Table \ref{tab:main} shows the overall performance comparison.
        We replaced the item encoding layer of three representative baselines, namely BPR++, SASRec++ and CARCA, and achieved performance improvement on almost all cases, as well as, all datasets in terms of NDCG.
        The final model, ProxyRCA reaches the state-of-the-art performance, improving up to 17\% when compared to the previously reported results.
        The improvement on BPR model demonstrates that our PIR layer is not bound to sequential models.
        Interesting point is that Fashion and Men were previously known to be sensitive to sequential information, but with PIR, BPR++ reaches comparable performance to sequential models.
        Note that in any case, since we are using $n_{\text{proxy}} = 0.5d$, only a minimum of 10\% and a maximum of 50\% of training parameters are used when compared to non-PIR counterparts.
        Moreover, our choice of the training objective, the LIP task, is generally superior to NIP task if tuned properly, which is shown in the appendix.
        
        To further analyze the above results, we counted the unique number of attributes and the number of duplicate items for each attribute, and implemented a 2-layered neural network that takes item attributes as input to perform a memorization task of classifying items based on their attributes.
        The result is summarized in Table \ref{tab:dataset}.
        In Fashion and Men, the item attributes distinguish the item accurately.
        In Game and ML-20M however, more than ten items share exactly the same attributes on average, and thus it is impractical to distinguish items by their attributes alone.
        Even with the datasets where attribute information is extremely weak, our ProxyRCA model can still outperform the baselines.
        We claim that the performance gain is due to the content-oriented shared embedding effect, similar to SSE \cite{wu2019stochastic}, explaining why the model outperforms even with limited attribute information.

\begin{figure}[tp]
    \centering
    \includegraphics[width=0.5\linewidth]{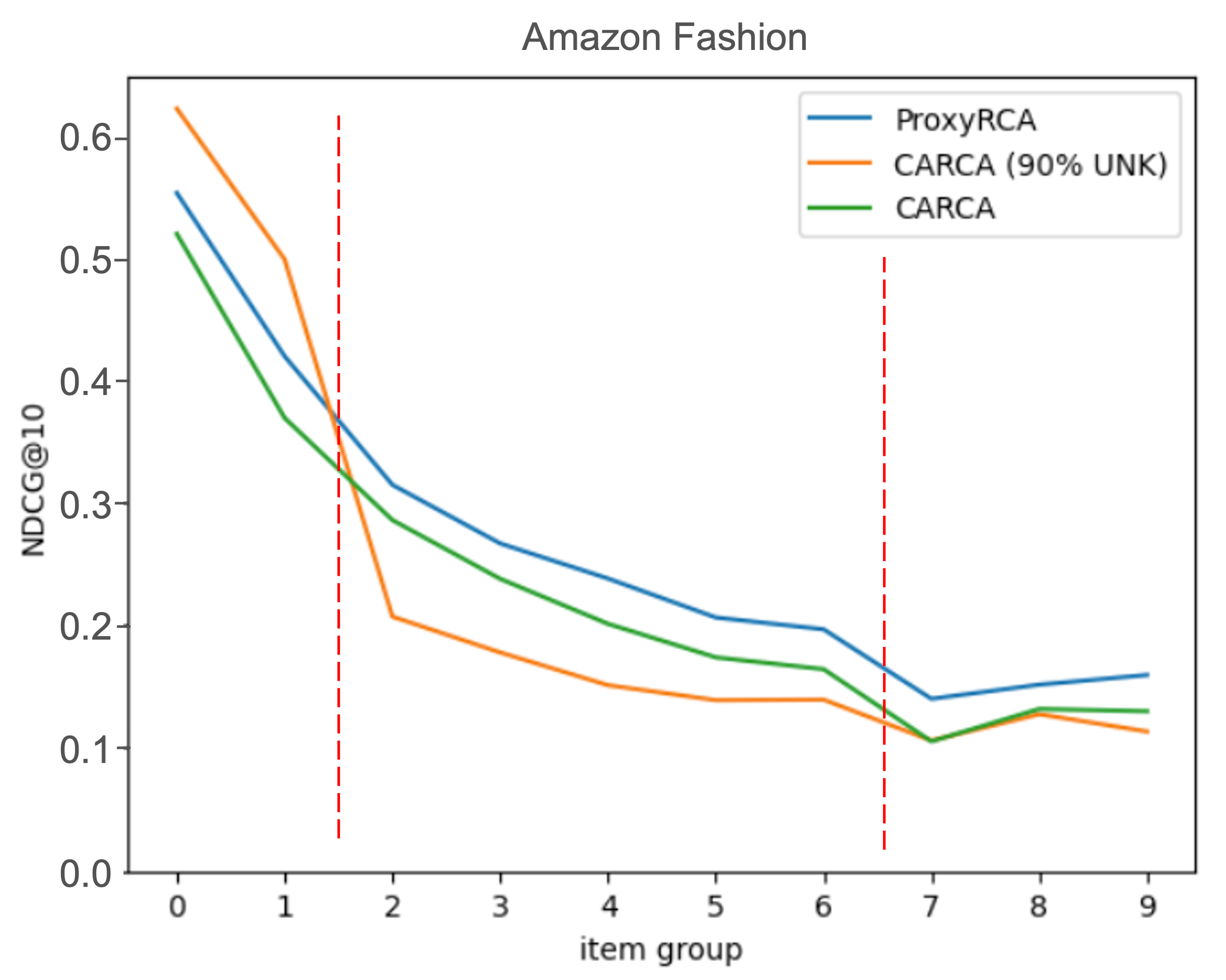}
    \caption{Performance comparison over frequency group.}
    \label{fig:ablation-group}
\end{figure}

    \subsection {Performance on Infrequent Items (RQ2)}

        As mentioned in Table \ref{tab:unk}, the baseline model showed a performance increase when infrequent item embeddings are replaced into the shared unknown token in sacrifice of the recommendation diversity.
        This phenomenon is not observed in models with PIR, which indirectly proves that the representation quality of the infrequent items is improved.

        For further investigation, we divide items into 10 groups, sorted by frequency, so that the sum of item occurrence in each group is equal (\ie the frequent group will have much fewer items compared to the infrequent group).
        We measure the mean performance of each group, average over each item on the test set, since an item can be a test item multiple times.
        The result is shown in Figure \ref{fig:ablation-group}.
        For item groups 0 to 1, the baseline model with 90\% of removal surmounts others but gives up the recommendation for majority of relatively infrequent item groups, which is an intuitive result.
        On the other hand, looking at item groups 7 to 9, CARCA does not show a big difference compared to the case of removing infrequent items, but ProxyRCA shows performance improvement in these groups as well, which also shows that the representation quality is improved.

    \subsection{Parameter Efficiency (RQ3)}

\begin{figure}[tp]
    \centering
    \setlength\abovecaptionskip{-0.0\baselineskip}
    \includegraphics[width=0.8\linewidth]{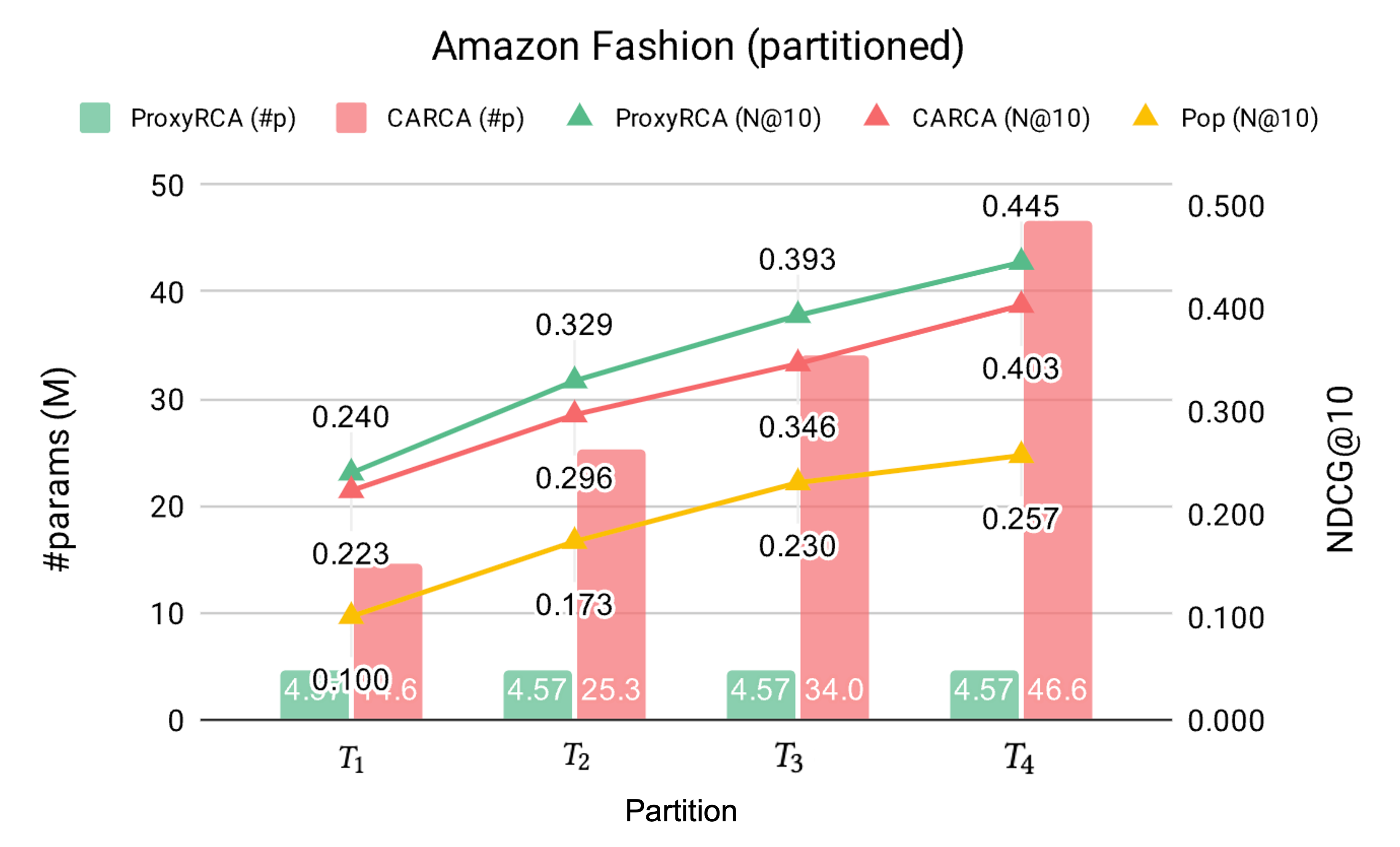}
    \caption{Parameter and performance comparison on the partitioned Fashion, simulating the growth of data.}
    \label{fig:param-fashion}
\end{figure}

\begin{figure}[tp]
    \centering
    \setlength\abovecaptionskip{-0.1\baselineskip}
    \includegraphics[width=0.9\linewidth]{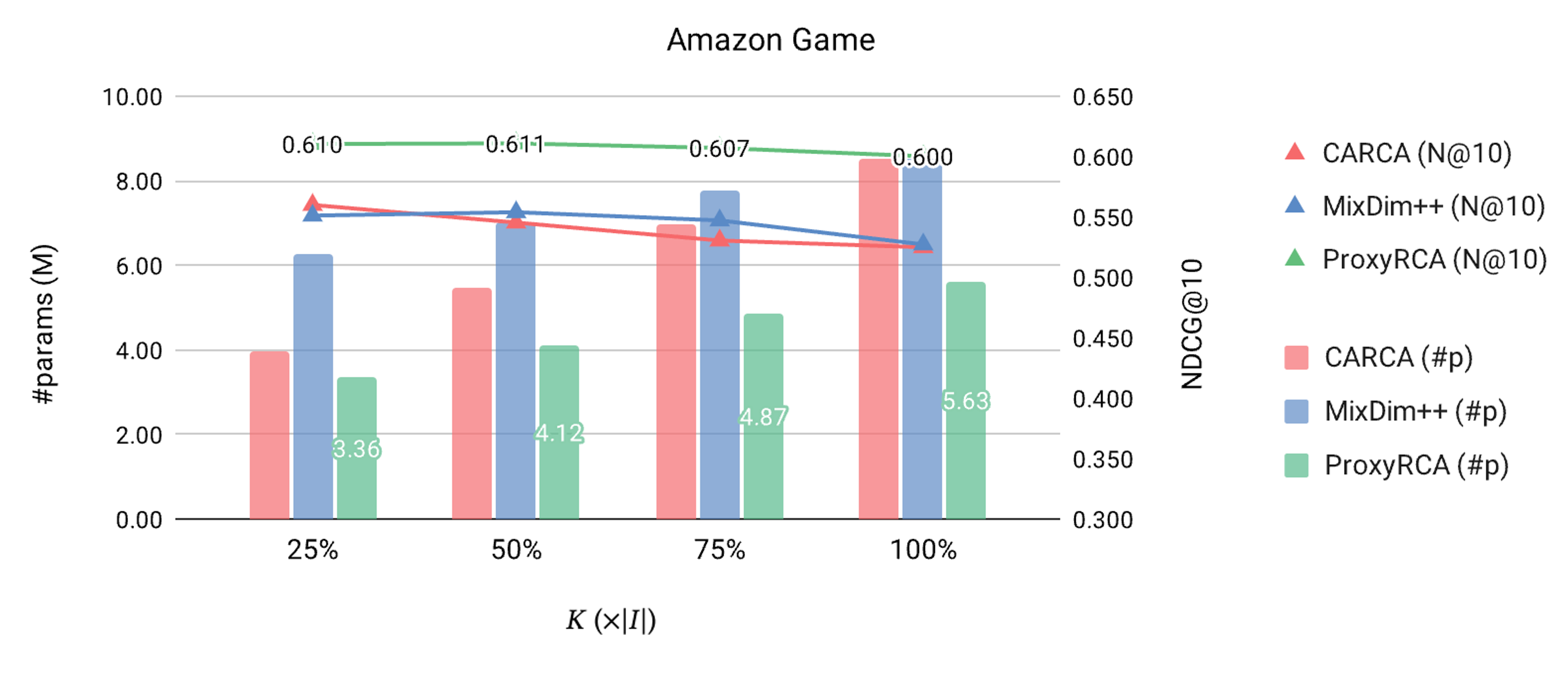}
    \caption{Parameter and performance comparison on Game over different $K$ value.}
    \label{fig:param-game}
\end{figure}

        To demonstrate the parameter efficiency of ProxyRCA, we created four new datasets, namely Fashion $T_{r}$ ($r = 1, 2, 3, 4$) by truncating Fashion so that the number of items for each partition is equal to $r/4 \times |I|$.
        These datasets simulate the real-world recommendation scenario where new items are added to the system.
        As shown in Figure \ref{fig:param-fashion}, the baseline model continuously expands the look-up table to match the increasing number of items as $r$ increases.
        However, ProxyRCA succeeds in maintaining a superior performance with a fixed number of parameters.

        To analyze the effect of the number of frequent items, we compare CARCA, MixDim++, and ProxyRCA in Game.
        For a fair comparison, CARCA treats items beyond top $K$ frequent items as unknown items, assigning the shared unknown embedding.
        For MixDim++, $K$ corresponds to the ratio of items using full latent dimension, while smaller dimension is set to $n_{\text{proxy}}$.
        The experiment result is in Figure \ref{fig:param-game}.
        We can see that even with similar settings, ProxyRCA uses less parameters and slow parameter growth rate, with superior recommendation performance.

\begin{figure}[tp]
    \centering
    \setlength\abovecaptionskip{0.3\baselineskip}
    \includegraphics[width=0.6\linewidth]{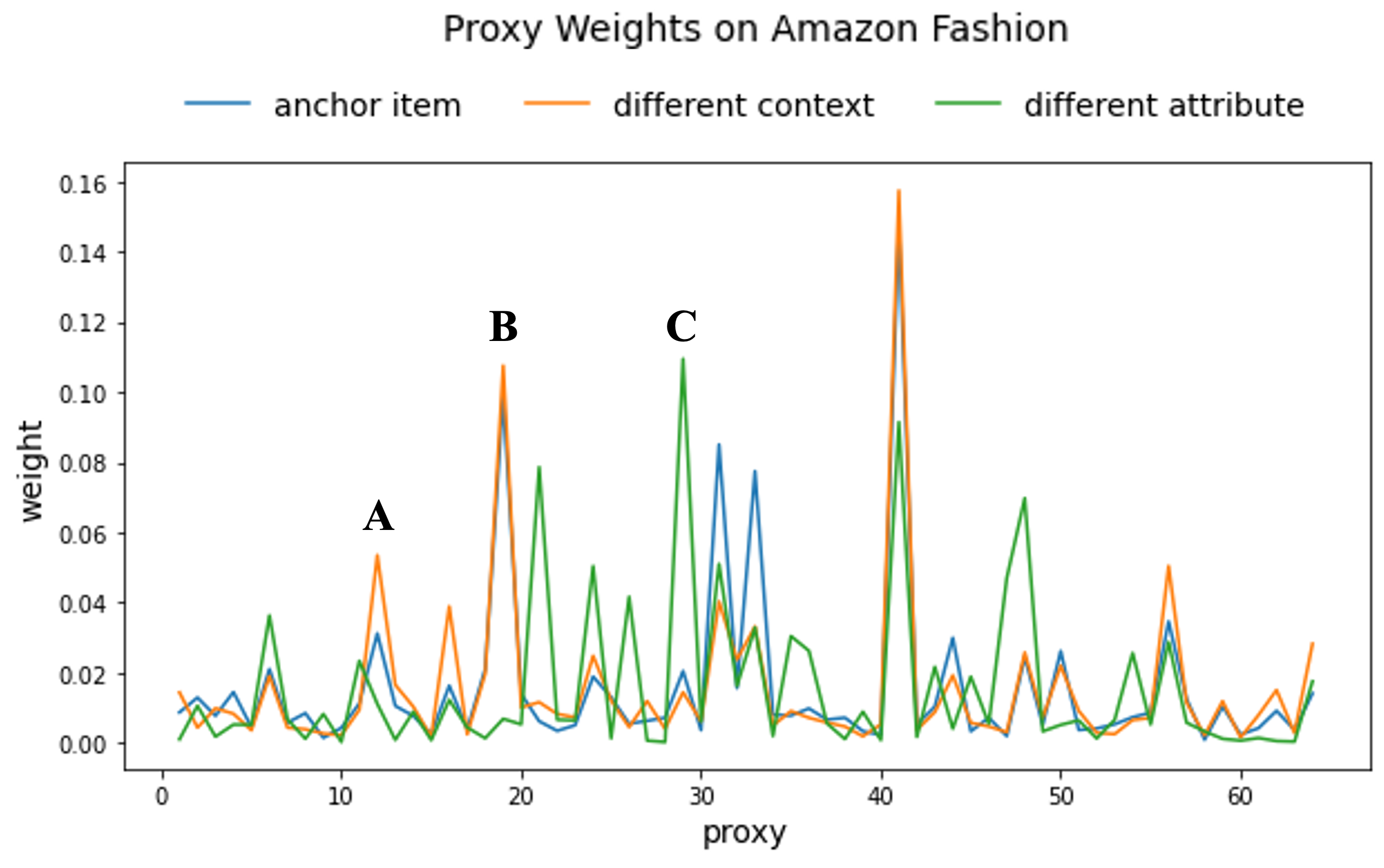}
    \caption{Visualization of proxy weights that reflect attribute and context modification.}
    \label{fig:viz-weight}
\end{figure}

\begin{figure}[tp]
    \centering
    \includegraphics[width=0.75\linewidth]{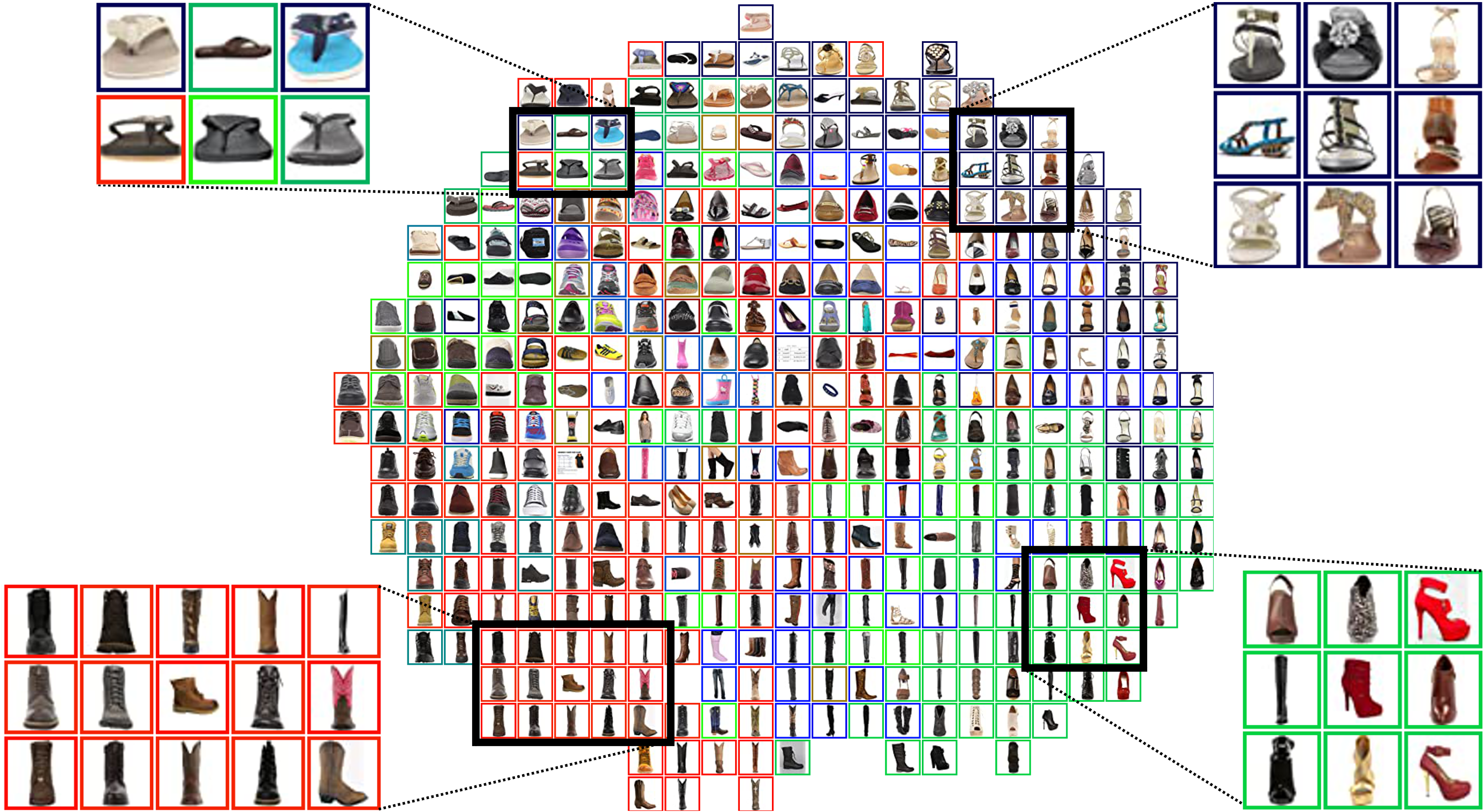}
    \caption{PCA visualization of item vectors on Fashion, using shoes category. The border colors denote the heaviest-weighted proxy.}
    \label{fig:viz-cluster}
\end{figure}

    \subsection{Proxy Analysis (RQ4)}
        
        Since the proxy weights are mainly computed using item attribute and context, the weight values should be different when the attribute and context are changed without the frequent item bias.
        In Fashion, we first randomly choose an anchor item and visualize the computed proxy weights, as in Figure \ref{fig:viz-weight}.
        To see the impact of context in proxy weights, we change context to an another value and visualize the weights.
        Similarly, to see the impact of attributes in proxy weights, we sample another item with different attributes but same context, and visualize the weights.
        As shown in the figure, due to the nature of Fashion that uses item images as attributes, the attributes are strongly reflected in the proxy weights.
        Therefore, we can see that most weights are similar when only contexts differ (points \textbf{B} and \textbf{C} in the figure).
        Even so, different context still assign different proxy weights (point \textbf{A} as demonstrated in the figure).
        This visualization indicates that each proxy has its role to represent items compositionally.

        To show that the proxy-based item representation can also provide a clustering effect on infrequent items, we visualize the item vectors with $K = 0$ using PCA, as in Figure \ref{fig:viz-cluster}.
        We use items of the shoes category among others (\eg dress, outer) in Fashion to see if the model successfully captures the subcategory information without any labels, which is a rather challenging task when compared to supercategory clustering.
        The figure shows that the items of similar subcategories (e.g., flip-flops on top-left, sandals on top-right, boots on bottom-left, and pumps on bottom-right in the figure) are naturally clustered, even without any explicit labels.


\section{Conclusion}

    In this paper, we highlighted the issues with the full-item embedding table from the perspective of infrequent item training through experimental demonstration.
    To overcome the issues, we proposed a proxy-based item representation model that can replace the existing item encoding layer in a plug-and-play fashion, by computing the item representation as a weighted sum of learnable proxy embeddings.
    To prove the effectiveness of our model, experiments were conducted using recommendation benchmark datasets and achieved state-of-the-art performance.


\begin{acks}
    This work was made in collaboration with Seoul National University and IntelliSys Co., Ltd.
    Also, this work was partly supported by the National Research Foundation of Korea (NRF) [No. RS-2023-00243243] and Institute of Information \& communications Technology Planning \& Evaluation (IITP) [NO.2021-0-01343, Artificial Intelligence Graduate School Program (Seoul National University)] [NO.2021-0-02068, Artificial Intelligence Innovation Hub (Artificial Intelligence Institute, Seoul National University)], both grant funded by the Korea government (MSIT).
\end{acks}


\bibliographystyle{reference}
\bibliography{reference}


\begin{thebibliography}{57}


\ifx \showCODEN    \undefined \def \showCODEN     #1{\unskip}     \fi
\ifx \showDOI      \undefined \def \showDOI       #1{#1}\fi
\ifx \showISBNx    \undefined \def \showISBNx     #1{\unskip}     \fi
\ifx \showISBNxiii \undefined \def \showISBNxiii  #1{\unskip}     \fi
\ifx \showISSN     \undefined \def \showISSN      #1{\unskip}     \fi
\ifx \showLCCN     \undefined \def \showLCCN      #1{\unskip}     \fi
\ifx \shownote     \undefined \def \shownote      #1{#1}          \fi
\ifx \showarticletitle \undefined \def \showarticletitle #1{#1}   \fi
\ifx \showURL      \undefined \def \showURL       {\relax}        \fi
\providecommand\bibfield[2]{#2}
\providecommand\bibinfo[2]{#2}
\providecommand\natexlab[1]{#1}
\providecommand\showeprint[2][]{arXiv:#2}

\bibitem[Ardalani et~al\mbox{.}(2022)]%
        {ardalani2022understanding}
\bibfield{author}{\bibinfo{person}{Newsha Ardalani}, \bibinfo{person}{Carole-Jean Wu}, \bibinfo{person}{Zeliang Chen}, \bibinfo{person}{Bhargav Bhushanam}, {and} \bibinfo{person}{Adnan Aziz}.} \bibinfo{year}{2022}\natexlab{}.
\newblock \showarticletitle{Understanding Scaling Laws for Recommendation Models}.
\newblock \bibinfo{journal}{\emph{arXiv preprint arXiv:2208.08489}} (\bibinfo{year}{2022}).
\newblock


\bibitem[Ba et~al\mbox{.}(2016)]%
        {ba2016layer}
\bibfield{author}{\bibinfo{person}{Jimmy~Lei Ba}, \bibinfo{person}{Jamie~Ryan Kiros}, {and} \bibinfo{person}{Geoffrey~E Hinton}.} \bibinfo{year}{2016}\natexlab{}.
\newblock \showarticletitle{Layer Normalization}.
\newblock \bibinfo{journal}{\emph{arXiv preprint arXiv:1607.06450}} (\bibinfo{year}{2016}).
\newblock


\bibitem[Batmaz et~al\mbox{.}(2019)]%
        {batmaz2019review}
\bibfield{author}{\bibinfo{person}{Zeynep Batmaz}, \bibinfo{person}{Ali Yurekli}, \bibinfo{person}{Alper Bilge}, {and} \bibinfo{person}{Cihan Kaleli}.} \bibinfo{year}{2019}\natexlab{}.
\newblock \showarticletitle{A Review on Deep Learning for Recommender Systems: Challenges and Remedies}.
\newblock \bibinfo{journal}{\emph{Artificial Intelligence Review}}  \bibinfo{volume}{52} (\bibinfo{year}{2019}), \bibinfo{pages}{1--37}.
\newblock


\bibitem[Caron et~al\mbox{.}(2020)]%
        {caron2020unsupervised}
\bibfield{author}{\bibinfo{person}{Mathilde Caron}, \bibinfo{person}{Ishan Misra}, \bibinfo{person}{Julien Mairal}, \bibinfo{person}{Priya Goyal}, \bibinfo{person}{Piotr Bojanowski}, {and} \bibinfo{person}{Armand Joulin}.} \bibinfo{year}{2020}\natexlab{}.
\newblock \showarticletitle{Unsupervised Learning of Visual Features by Contrasting Cluster Assignments}.
\newblock \bibinfo{journal}{\emph{Advances in Neural Information Processing Systems}}  \bibinfo{volume}{33} (\bibinfo{year}{2020}), \bibinfo{pages}{9912--9924}.
\newblock


\bibitem[Caron et~al\mbox{.}(2021)]%
        {caron2021emerging}
\bibfield{author}{\bibinfo{person}{Mathilde Caron}, \bibinfo{person}{Hugo Touvron}, \bibinfo{person}{Ishan Misra}, \bibinfo{person}{Herv{\'e} J{\'e}gou}, \bibinfo{person}{Julien Mairal}, \bibinfo{person}{Piotr Bojanowski}, {and} \bibinfo{person}{Armand Joulin}.} \bibinfo{year}{2021}\natexlab{}.
\newblock \showarticletitle{Emerging Properties in Self-supervised Vision Transformers}. In \bibinfo{booktitle}{\emph{Proceedings of the IEEE/CVF International Conference on Computer Vision}}. \bibinfo{pages}{9650--9660}.
\newblock


\bibitem[Chen et~al\mbox{.}(2022b)]%
        {chen2022generative}
\bibfield{author}{\bibinfo{person}{Hao Chen}, \bibinfo{person}{Zefan Wang}, \bibinfo{person}{Feiran Huang}, \bibinfo{person}{Xiao Huang}, \bibinfo{person}{Yue Xu}, \bibinfo{person}{Yishi Lin}, \bibinfo{person}{Peng He}, {and} \bibinfo{person}{Zhoujun Li}.} \bibinfo{year}{2022}\natexlab{b}.
\newblock \showarticletitle{Generative Adversarial Framework for Cold-start Item Recommendation}. In \bibinfo{booktitle}{\emph{Proceedings of the 45th International ACM SIGIR Conference on Research and Development in Information Retrieval}}. \bibinfo{pages}{2565--2571}.
\newblock


\bibitem[Chen et~al\mbox{.}(2020)]%
        {chen2020simple}
\bibfield{author}{\bibinfo{person}{Ting Chen}, \bibinfo{person}{Simon Kornblith}, \bibinfo{person}{Mohammad Norouzi}, {and} \bibinfo{person}{Geoffrey Hinton}.} \bibinfo{year}{2020}\natexlab{}.
\newblock \showarticletitle{A Simple Framework for Contrastive Learning of Visual Representations}. In \bibinfo{booktitle}{\emph{International Conference on Machine Learning}}. \bibinfo{pages}{1597--1607}.
\newblock


\bibitem[Chen et~al\mbox{.}(2023)]%
        {chen2023clustered}
\bibfield{author}{\bibinfo{person}{Yizhou Chen}, \bibinfo{person}{Guangda Huzhang}, \bibinfo{person}{Anxiang Zeng}, \bibinfo{person}{Qingtao Yu}, \bibinfo{person}{Hui Sun}, \bibinfo{person}{Hengyi Li}, \bibinfo{person}{Jingyi Li}, \bibinfo{person}{Yabo Ni}, \bibinfo{person}{Han Yu}, {and} \bibinfo{person}{Zhiming Zhou}.} \bibinfo{year}{2023}\natexlab{}.
\newblock \showarticletitle{Clustered Embedding Learning for Recommender Systems}. In \bibinfo{booktitle}{\emph{The World Wide Web Conference}}.
\newblock


\bibitem[Chen et~al\mbox{.}(2022a)]%
        {chen2022intent}
\bibfield{author}{\bibinfo{person}{Yongjun Chen}, \bibinfo{person}{Zhiwei Liu}, \bibinfo{person}{Jia Li}, \bibinfo{person}{Julian McAuley}, {and} \bibinfo{person}{Caiming Xiong}.} \bibinfo{year}{2022}\natexlab{a}.
\newblock \showarticletitle{Intent Contrastive Learning for Sequential Recommendation}.
\newblock \bibinfo{journal}{\emph{arXiv preprint arXiv:2202.02519}} (\bibinfo{year}{2022}).
\newblock


\bibitem[Cho et~al\mbox{.}(2021b)]%
        {cho2021masked}
\bibfield{author}{\bibinfo{person}{Hyunsoo Cho}, \bibinfo{person}{Jinseok Seol}, {and} \bibinfo{person}{Sang-goo Lee}.} \bibinfo{year}{2021}\natexlab{b}.
\newblock \showarticletitle{Masked Contrastive Learning for Anomaly Detection}.
\newblock \bibinfo{journal}{\emph{The 30th International Joint Conference on Artificial Intelligence}} (\bibinfo{year}{2021}).
\newblock


\bibitem[Cho et~al\mbox{.}(2021a)]%
        {cho2021unsupervised}
\bibfield{author}{\bibinfo{person}{Junsu Cho}, \bibinfo{person}{SeongKu Kang}, \bibinfo{person}{Dongmin Hyun}, {and} \bibinfo{person}{Hwanjo Yu}.} \bibinfo{year}{2021}\natexlab{a}.
\newblock \showarticletitle{Unsupervised Proxy Selection for Session-based Recommender Systems}. In \bibinfo{booktitle}{\emph{Proceedings of the 44th International ACM SIGIR Conference on Research and Development in Information Retrieval}}. \bibinfo{pages}{327--336}.
\newblock


\bibitem[Covington et~al\mbox{.}(2016)]%
        {covington2016deep}
\bibfield{author}{\bibinfo{person}{Paul Covington}, \bibinfo{person}{Jay Adams}, {and} \bibinfo{person}{Emre Sargin}.} \bibinfo{year}{2016}\natexlab{}.
\newblock \showarticletitle{Deep Neural Networks for Youtube Recommendations}. In \bibinfo{booktitle}{\emph{Proceedings of the 10th ACM Conference on Recommender Systems}}. \bibinfo{pages}{191--198}.
\newblock


\bibitem[Dong et~al\mbox{.}(2020)]%
        {dong2020mamo}
\bibfield{author}{\bibinfo{person}{Manqing Dong}, \bibinfo{person}{Feng Yuan}, \bibinfo{person}{Lina Yao}, \bibinfo{person}{Xiwei Xu}, {and} \bibinfo{person}{Liming Zhu}.} \bibinfo{year}{2020}\natexlab{}.
\newblock \showarticletitle{Mamo: Memory-augmented Meta-optimization for Cold-start Recommendation}. In \bibinfo{booktitle}{\emph{Proceedings of the 26th ACM SIGKDD International Conference on Knowledge Discovery \& Data Mining}}. \bibinfo{pages}{688--697}.
\newblock


\bibitem[Fan et~al\mbox{.}(2021)]%
        {fan2021continuous}
\bibfield{author}{\bibinfo{person}{Ziwei Fan}, \bibinfo{person}{Zhiwei Liu}, \bibinfo{person}{Jiawei Zhang}, \bibinfo{person}{Yun Xiong}, \bibinfo{person}{Lei Zheng}, {and} \bibinfo{person}{Philip~S Yu}.} \bibinfo{year}{2021}\natexlab{}.
\newblock \showarticletitle{Continuous-time Sequential Recommendation with Temporal Graph Collaborative Transformer}. In \bibinfo{booktitle}{\emph{Proceedings of the 30th ACM International Conference on Information \& Knowledge Management}}. \bibinfo{pages}{433--442}.
\newblock


\bibitem[Guo et~al\mbox{.}(2017)]%
        {guo2017deepfm}
\bibfield{author}{\bibinfo{person}{Huifeng Guo}, \bibinfo{person}{Ruiming Tang}, \bibinfo{person}{Yunming Ye}, \bibinfo{person}{Zhenguo Li}, {and} \bibinfo{person}{Xiuqiang He}.} \bibinfo{year}{2017}\natexlab{}.
\newblock \showarticletitle{DeepFM: A Factorization-machine based Neural Network for CTR Prediction}.
\newblock \bibinfo{journal}{\emph{Proceedings of the 26th International Joint Conference on Artificial Intelligence}} (\bibinfo{year}{2017}).
\newblock


\bibitem[He and McAuley(2016)]%
        {he2016vbpr}
\bibfield{author}{\bibinfo{person}{Ruining He} {and} \bibinfo{person}{Julian McAuley}.} \bibinfo{year}{2016}\natexlab{}.
\newblock \showarticletitle{VBPR: Visual Bayesian Personalized Ranking from Implicit Feedback}. In \bibinfo{booktitle}{\emph{Proceedings of the AAAI conference on Artificial Intelligence}}, Vol.~\bibinfo{volume}{30}.
\newblock


\bibitem[He et~al\mbox{.}(2017)]%
        {he2017neural}
\bibfield{author}{\bibinfo{person}{Xiangnan He}, \bibinfo{person}{Lizi Liao}, \bibinfo{person}{Hanwang Zhang}, \bibinfo{person}{Liqiang Nie}, \bibinfo{person}{Xia Hu}, {and} \bibinfo{person}{Tat-Seng Chua}.} \bibinfo{year}{2017}\natexlab{}.
\newblock \showarticletitle{Neural Collaborative Filtering}. In \bibinfo{booktitle}{\emph{Proceedings of the 26th International Conference on World Wide Web}}. \bibinfo{pages}{173--182}.
\newblock


\bibitem[Kang et~al\mbox{.}(2021)]%
        {kang2021learning}
\bibfield{author}{\bibinfo{person}{Wang-Cheng Kang}, \bibinfo{person}{Derek~Zhiyuan Cheng}, \bibinfo{person}{Tiansheng Yao}, \bibinfo{person}{Xinyang Yi}, \bibinfo{person}{Ting Chen}, \bibinfo{person}{Lichan Hong}, {and} \bibinfo{person}{Ed~H Chi}.} \bibinfo{year}{2021}\natexlab{}.
\newblock \showarticletitle{Learning to Embed Categorical Features without Embedding Tables for Recommendation}. In \bibinfo{booktitle}{\emph{Proceedings of the 27th ACM SIGKDD Conference on Knowledge Discovery \& Data Mining}}. \bibinfo{pages}{840--850}.
\newblock


\bibitem[Kang and McAuley(2018)]%
        {kang2018self}
\bibfield{author}{\bibinfo{person}{Wang-Cheng Kang} {and} \bibinfo{person}{Julian McAuley}.} \bibinfo{year}{2018}\natexlab{}.
\newblock \showarticletitle{Self-attentive Sequential Recommendation}. In \bibinfo{booktitle}{\emph{2018 IEEE International Conference on Data Mining}}. \bibinfo{pages}{197--206}.
\newblock


\bibitem[Kim et~al\mbox{.}(2020)]%
        {kim2020proxy}
\bibfield{author}{\bibinfo{person}{Sungyeon Kim}, \bibinfo{person}{Dongwon Kim}, \bibinfo{person}{Minsu Cho}, {and} \bibinfo{person}{Suha Kwak}.} \bibinfo{year}{2020}\natexlab{}.
\newblock \showarticletitle{Proxy Anchor Koss for Deep Metric Learning}. In \bibinfo{booktitle}{\emph{Proceedings of the IEEE/CVF Conference on Computer Vision and Pattern Recognition}}. \bibinfo{pages}{3238--3247}.
\newblock


\bibitem[Koren et~al\mbox{.}(2009)]%
        {koren2009matrix}
\bibfield{author}{\bibinfo{person}{Yehuda Koren}, \bibinfo{person}{Robert Bell}, {and} \bibinfo{person}{Chris Volinsky}.} \bibinfo{year}{2009}\natexlab{}.
\newblock \showarticletitle{Matrix Factorization Techniques for Recommender Systems}.
\newblock \bibinfo{journal}{\emph{Computer}} \bibinfo{volume}{42}, \bibinfo{number}{8} (\bibinfo{year}{2009}), \bibinfo{pages}{30--37}.
\newblock


\bibitem[Li et~al\mbox{.}(2021b)]%
        {li2020prototypical}
\bibfield{author}{\bibinfo{person}{Junnan Li}, \bibinfo{person}{Pan Zhou}, \bibinfo{person}{Caiming Xiong}, {and} \bibinfo{person}{Steven~CH Hoi}.} \bibinfo{year}{2021}\natexlab{b}.
\newblock \showarticletitle{Prototypical Contrastive Learning of Unsupervised Representations}.
\newblock \bibinfo{journal}{\emph{International Conference on Learning Representations}} (\bibinfo{year}{2021}).
\newblock


\bibitem[Li et~al\mbox{.}(2021a)]%
        {li2021lightweight}
\bibfield{author}{\bibinfo{person}{Yang Li}, \bibinfo{person}{Tong Chen}, \bibinfo{person}{Peng-Fei Zhang}, {and} \bibinfo{person}{Hongzhi Yin}.} \bibinfo{year}{2021}\natexlab{a}.
\newblock \showarticletitle{Lightweight Self-attentive Sequential Recommendation}. In \bibinfo{booktitle}{\emph{Proceedings of the 30th ACM International Conference on Information \& Knowledge Management}}. \bibinfo{pages}{967--977}.
\newblock


\bibitem[Liu et~al\mbox{.}(2021)]%
        {liu2021learnable}
\bibfield{author}{\bibinfo{person}{Siyi Liu}, \bibinfo{person}{Chen Gao}, \bibinfo{person}{Yihong Chen}, \bibinfo{person}{Depeng Jin}, {and} \bibinfo{person}{Yong Li}.} \bibinfo{year}{2021}\natexlab{}.
\newblock \showarticletitle{Learnable embedding sizes for recommender systems}.
\newblock \bibinfo{journal}{\emph{2021 International Conference on Learning Representations}} (\bibinfo{year}{2021}).
\newblock


\bibitem[Loshchilov and Hutter(2019)]%
        {loshchilov2018fixing}
\bibfield{author}{\bibinfo{person}{Ilya Loshchilov} {and} \bibinfo{person}{Frank Hutter}.} \bibinfo{year}{2019}\natexlab{}.
\newblock \showarticletitle{Decoupled Weight Decay Regularization}.
\newblock \bibinfo{journal}{\emph{2019 International Conference on Learning Representations}} (\bibinfo{year}{2019}).
\newblock


\bibitem[McAuley et~al\mbox{.}(2015)]%
        {mcauley2015image}
\bibfield{author}{\bibinfo{person}{Julian McAuley}, \bibinfo{person}{Christopher Targett}, \bibinfo{person}{Qinfeng Shi}, {and} \bibinfo{person}{Anton Van Den~Hengel}.} \bibinfo{year}{2015}\natexlab{}.
\newblock \showarticletitle{Image-based Recommendations on Styles and Substitutes}. In \bibinfo{booktitle}{\emph{Proceedings of the 38th International ACM SIGIR Conference on Research and Development in Information Retrieval}}. \bibinfo{pages}{43--52}.
\newblock


\bibitem[Mikolov et~al\mbox{.}(2013)]%
        {mikolov2013distributed}
\bibfield{author}{\bibinfo{person}{Tomas Mikolov}, \bibinfo{person}{Ilya Sutskever}, \bibinfo{person}{Kai Chen}, \bibinfo{person}{Greg~S Corrado}, {and} \bibinfo{person}{Jeff Dean}.} \bibinfo{year}{2013}\natexlab{}.
\newblock \showarticletitle{Distributed Representations of Words and Phrases and their Compositionality}.
\newblock \bibinfo{journal}{\emph{Advances in Neural Information Processing Systems}}  \bibinfo{volume}{26} (\bibinfo{year}{2013}).
\newblock


\bibitem[Movshovitz-Attias et~al\mbox{.}(2017)]%
        {movshovitz2017no}
\bibfield{author}{\bibinfo{person}{Yair Movshovitz-Attias}, \bibinfo{person}{Alexander Toshev}, \bibinfo{person}{Thomas~K Leung}, \bibinfo{person}{Sergey Ioffe}, {and} \bibinfo{person}{Saurabh Singh}.} \bibinfo{year}{2017}\natexlab{}.
\newblock \showarticletitle{No Fuss Distance Metric Learning using Proxies}. In \bibinfo{booktitle}{\emph{Proceedings of the IEEE International Conference on Computer Vision}}. \bibinfo{pages}{360--368}.
\newblock


\bibitem[Naumov et~al\mbox{.}(2021)]%
        {ginart2021mixed}
\bibfield{author}{\bibinfo{person}{Maxim Naumov}, \bibinfo{person}{Dheevatsa Mudigere}, \bibinfo{person}{Jiyan Yang}, {and} \bibinfo{person}{James Zou}.} \bibinfo{year}{2021}\natexlab{}.
\newblock \showarticletitle{Mixed Dimension Embeddings with Application to Memory-efficient Recommendation Systems}. In \bibinfo{booktitle}{\emph{IEEE International Symposium on Information Theory}}. \bibinfo{pages}{2786--2791}.
\newblock


\bibitem[Ning and Karypis(2011)]%
        {ning2011slim}
\bibfield{author}{\bibinfo{person}{Xia Ning} {and} \bibinfo{person}{George Karypis}.} \bibinfo{year}{2011}\natexlab{}.
\newblock \showarticletitle{SLIM: Sparse Linear Methods for Top-n Recommender Systems}. In \bibinfo{booktitle}{\emph{IEEE 11th International Conference on Data Mining}}. \bibinfo{pages}{497--506}.
\newblock


\bibitem[Petrov and Macdonald(2022a)]%
        {petrov2022effective}
\bibfield{author}{\bibinfo{person}{Aleksandr Petrov} {and} \bibinfo{person}{Craig Macdonald}.} \bibinfo{year}{2022}\natexlab{a}.
\newblock \showarticletitle{Effective and Efficient Training for Sequential Recommendation using Recency Sampling}. In \bibinfo{booktitle}{\emph{Proceedings of the 16th ACM Conference on Recommender Systems}}. \bibinfo{pages}{81--91}.
\newblock


\bibitem[Petrov and Macdonald(2022b)]%
        {petrov2022systematic}
\bibfield{author}{\bibinfo{person}{Aleksandr Petrov} {and} \bibinfo{person}{Craig Macdonald}.} \bibinfo{year}{2022}\natexlab{b}.
\newblock \showarticletitle{A Systematic Review and Replicability Study of BERT4Rec for Sequential Recommendation}. In \bibinfo{booktitle}{\emph{Proceedings of the 16th ACM Conference on Recommender Systems}}. \bibinfo{pages}{436--447}.
\newblock


\bibitem[Pfadler et~al\mbox{.}(2020)]%
        {pfadler2020billion}
\bibfield{author}{\bibinfo{person}{Andreas Pfadler}, \bibinfo{person}{Huan Zhao}, \bibinfo{person}{Jizhe Wang}, \bibinfo{person}{Lifeng Wang}, \bibinfo{person}{Pipei Huang}, {and} \bibinfo{person}{Dik~Lun Lee}.} \bibinfo{year}{2020}\natexlab{}.
\newblock \showarticletitle{Billion-scale recommendation with Heterogeneous Side Information at Taobao}. In \bibinfo{booktitle}{\emph{IEEE 36th International Conference on Data Engineering}}. \bibinfo{pages}{1667--1676}.
\newblock


\bibitem[Quadrana et~al\mbox{.}(2017)]%
        {quadrana2017personalizing}
\bibfield{author}{\bibinfo{person}{Massimo Quadrana}, \bibinfo{person}{Alexandros Karatzoglou}, \bibinfo{person}{Bal{\'a}zs Hidasi}, {and} \bibinfo{person}{Paolo Cremonesi}.} \bibinfo{year}{2017}\natexlab{}.
\newblock \showarticletitle{Personalizing Session-based Recommendations with Hierarchical Recurrent Neural Networks}. In \bibinfo{booktitle}{\emph{Proceedings of the 11th ACM Conference on Recommender Systems}}. \bibinfo{pages}{130--137}.
\newblock


\bibitem[Rashed et~al\mbox{.}(2022)]%
        {rashed2022context}
\bibfield{author}{\bibinfo{person}{Ahmed Rashed}, \bibinfo{person}{Shereen Elsayed}, {and} \bibinfo{person}{Lars Schmidt-Thieme}.} \bibinfo{year}{2022}\natexlab{}.
\newblock \showarticletitle{Context and Attribute-Aware Sequential Recommendation via Cross-Attention}. In \bibinfo{booktitle}{\emph{Proceedings of the 16th ACM Conference on Recommender Systems}}. \bibinfo{pages}{71--80}.
\newblock


\bibitem[Rendle(2010)]%
        {rendle2010factorization}
\bibfield{author}{\bibinfo{person}{Steffen Rendle}.} \bibinfo{year}{2010}\natexlab{}.
\newblock \showarticletitle{Factorization Machines}. In \bibinfo{booktitle}{\emph{IEEE International Conference on Data Mining}}. \bibinfo{pages}{995--1000}.
\newblock


\bibitem[Rendle et~al\mbox{.}(2009)]%
        {rendle2009bpr}
\bibfield{author}{\bibinfo{person}{Steffen Rendle}, \bibinfo{person}{Christoph Freudenthaler}, \bibinfo{person}{Zeno Gantner}, {and} \bibinfo{person}{Lars Schmidt-Thieme}.} \bibinfo{year}{2009}\natexlab{}.
\newblock \showarticletitle{BPR: Bayesian Personalized Ranking from Implicit Feedback}. In \bibinfo{booktitle}{\emph{Proceedings of the 25th Conference on Uncertainty in Artificial Intelligence}}. \bibinfo{pages}{452--461}.
\newblock


\bibitem[Sarwar et~al\mbox{.}(2001)]%
        {sarwar2001item}
\bibfield{author}{\bibinfo{person}{Badrul Sarwar}, \bibinfo{person}{George Karypis}, \bibinfo{person}{Joseph Konstan}, {and} \bibinfo{person}{John Riedl}.} \bibinfo{year}{2001}\natexlab{}.
\newblock \showarticletitle{Item-based Collaborative Filtering Recommendation Algorithms}. In \bibinfo{booktitle}{\emph{Proceedings of the 10th International Conference on World Wide Web}}. \bibinfo{pages}{285--295}.
\newblock


\bibitem[Seol et~al\mbox{.}(2022)]%
        {seol2022exploiting}
\bibfield{author}{\bibinfo{person}{Jinseok Seol}, \bibinfo{person}{Youngrok Ko}, {and} \bibinfo{person}{Sang-goo Lee}.} \bibinfo{year}{2022}\natexlab{}.
\newblock \showarticletitle{Exploiting Session Information in BERT-based Session-aware Sequential Recommendation}. In \bibinfo{booktitle}{\emph{Proceedings of the 45th International ACM SIGIR Conference on Research and Development in Information Retrieval}}. \bibinfo{pages}{2639--2644}.
\newblock


\bibitem[Shen et~al\mbox{.}(2021)]%
        {shen2021umec}
\bibfield{author}{\bibinfo{person}{Jiayi Shen}, \bibinfo{person}{Haotao Wang}, \bibinfo{person}{Shupeng Gui}, \bibinfo{person}{Jianchao Tan}, \bibinfo{person}{Zhangyang Wang}, {and} \bibinfo{person}{Ji Liu}.} \bibinfo{year}{2021}\natexlab{}.
\newblock \showarticletitle{UMEC: Unified Model and Embedding Compression for Efficient Recommendation Systems}.
\newblock \bibinfo{journal}{\emph{International Conference on Learning Representations}}.
\newblock


\bibitem[Sohn(2016)]%
        {sohn2016improved}
\bibfield{author}{\bibinfo{person}{Kihyuk Sohn}.} \bibinfo{year}{2016}\natexlab{}.
\newblock \showarticletitle{Improved Deep Metric Learning with Multi-class n-pair Loss Objective}.
\newblock \bibinfo{journal}{\emph{Advances in Neural Information Processing Systems}}  \bibinfo{volume}{29} (\bibinfo{year}{2016}).
\newblock


\bibitem[Sun et~al\mbox{.}(2019)]%
        {sun2019bert4rec}
\bibfield{author}{\bibinfo{person}{Fei Sun}, \bibinfo{person}{Jun Liu}, \bibinfo{person}{Jian Wu}, \bibinfo{person}{Changhua Pei}, \bibinfo{person}{Xiao Lin}, \bibinfo{person}{Wenwu Ou}, {and} \bibinfo{person}{Peng Jiang}.} \bibinfo{year}{2019}\natexlab{}.
\newblock \showarticletitle{BERT4Rec: Sequential Recommendation with Bidirectional Encoder Representations from Transformer}. In \bibinfo{booktitle}{\emph{Proceedings of the 28th ACM International Conference on Information and Knowledge Management}}. \bibinfo{pages}{1441--1450}.
\newblock


\bibitem[Tran et~al\mbox{.}(2019)]%
        {tran2019improving}
\bibfield{author}{\bibinfo{person}{Viet-Anh Tran}, \bibinfo{person}{Romain Hennequin}, \bibinfo{person}{Jimena Royo-Letelier}, {and} \bibinfo{person}{Manuel Moussallam}.} \bibinfo{year}{2019}\natexlab{}.
\newblock \showarticletitle{Improving Collaborative Metric Learning with Efficient Negative Sampling}. In \bibinfo{booktitle}{\emph{Proceedings of the 42nd International ACM SIGIR Conference on Research and Development in Information Retrieval}}. \bibinfo{pages}{1201--1204}.
\newblock


\bibitem[Vaswani et~al\mbox{.}(2017)]%
        {vaswani2017attention}
\bibfield{author}{\bibinfo{person}{Ashish Vaswani}, \bibinfo{person}{Noam Shazeer}, \bibinfo{person}{Niki Parmar}, \bibinfo{person}{Jakob Uszkoreit}, \bibinfo{person}{Llion Jones}, \bibinfo{person}{Aidan~N Gomez}, \bibinfo{person}{{\L}ukasz Kaiser}, {and} \bibinfo{person}{Illia Polosukhin}.} \bibinfo{year}{2017}\natexlab{}.
\newblock \showarticletitle{Attention is All You Need}.
\newblock \bibinfo{journal}{\emph{Advances in Neural Information Processing Systems}}  \bibinfo{volume}{30} (\bibinfo{year}{2017}).
\newblock


\bibitem[Wang et~al\mbox{.}(2019)]%
        {wang2019sequential}
\bibfield{author}{\bibinfo{person}{Shoujin Wang}, \bibinfo{person}{Liang Hu}, \bibinfo{person}{Yan Wang}, \bibinfo{person}{Longbing Cao}, \bibinfo{person}{Quan~Z Sheng}, {and} \bibinfo{person}{Mehmet Orgun}.} \bibinfo{year}{2019}\natexlab{}.
\newblock \showarticletitle{Sequential Recommender Systems: Challenges, Progress and Prospects}.
\newblock \bibinfo{journal}{\emph{Proceedings of the 28th International Joint Conference on Artificial Intelligence}} (\bibinfo{year}{2019}).
\newblock


\bibitem[Wei et~al\mbox{.}(2021)]%
        {wei2021contrastive}
\bibfield{author}{\bibinfo{person}{Yinwei Wei}, \bibinfo{person}{Xiang Wang}, \bibinfo{person}{Qi Li}, \bibinfo{person}{Liqiang Nie}, \bibinfo{person}{Yan Li}, \bibinfo{person}{Xuanping Li}, {and} \bibinfo{person}{Tat-Seng Chua}.} \bibinfo{year}{2021}\natexlab{}.
\newblock \showarticletitle{Contrastive Learning for Cold-start Recommendation}. In \bibinfo{booktitle}{\emph{Proceedings of the 29th ACM International Conference on Multimedia}}. \bibinfo{pages}{5382--5390}.
\newblock


\bibitem[Wieczorek et~al\mbox{.}(2021)]%
        {wieczorek2021unreasonable}
\bibfield{author}{\bibinfo{person}{Miko{\l}aj Wieczorek}, \bibinfo{person}{Barbara Rychalska}, {and} \bibinfo{person}{Jacek D{\k{a}}browski}.} \bibinfo{year}{2021}\natexlab{}.
\newblock \showarticletitle{On the Unreasonable Effectiveness of Centroids in Image Retrieval}. In \bibinfo{booktitle}{\emph{Neural Information Processing: 28th International Conference, ICONIP 2021, Proceedings, Part IV 28}}. Springer, \bibinfo{pages}{212--223}.
\newblock


\bibitem[Wu et~al\mbox{.}(2020)]%
        {wu2020sse}
\bibfield{author}{\bibinfo{person}{Liwei Wu}, \bibinfo{person}{Shuqing Li}, \bibinfo{person}{Cho-Jui Hsieh}, {and} \bibinfo{person}{James Sharpnack}.} \bibinfo{year}{2020}\natexlab{}.
\newblock \showarticletitle{SSE-PT: Sequential Recommendation via Personalized Transformer}. In \bibinfo{booktitle}{\emph{Proceedings of the 14th ACM Conference on Recommender Systems}}. \bibinfo{pages}{328--337}.
\newblock


\bibitem[Wu et~al\mbox{.}(2019)]%
        {wu2019stochastic}
\bibfield{author}{\bibinfo{person}{Liwei Wu}, \bibinfo{person}{Shuqing Li}, \bibinfo{person}{Cho-Jui Hsieh}, {and} \bibinfo{person}{James~L Sharpnack}.} \bibinfo{year}{2019}\natexlab{}.
\newblock \showarticletitle{Stochastic Shared Embeddings: Data-driven Regularization of Embedding Layers}.
\newblock \bibinfo{journal}{\emph{Advances in Neural Information Processing Systems}}  \bibinfo{volume}{32} (\bibinfo{year}{2019}).
\newblock


\bibitem[Xue et~al\mbox{.}(2017)]%
        {xue2017deep}
\bibfield{author}{\bibinfo{person}{Hong-Jian Xue}, \bibinfo{person}{Xinyu Dai}, \bibinfo{person}{Jianbing Zhang}, \bibinfo{person}{Shujian Huang}, {and} \bibinfo{person}{Jiajun Chen}.} \bibinfo{year}{2017}\natexlab{}.
\newblock \showarticletitle{Deep Matrix Factorization Models for Recommender Systems.}. In \bibinfo{booktitle}{\emph{IJCAI}}, Vol.~\bibinfo{volume}{17}. Melbourne, Australia, \bibinfo{pages}{3203--3209}.
\newblock


\bibitem[Yan et~al\mbox{.}(2021)]%
        {yan2021learning}
\bibfield{author}{\bibinfo{person}{Bencheng Yan}, \bibinfo{person}{Pengjie Wang}, \bibinfo{person}{Kai Zhang}, \bibinfo{person}{Wei Lin}, \bibinfo{person}{Kuang-Chih Lee}, \bibinfo{person}{Jian Xu}, {and} \bibinfo{person}{Bo Zheng}.} \bibinfo{year}{2021}\natexlab{}.
\newblock \showarticletitle{Learning Effective and Efficient Embedding via an Adaptively-masked Twins-based Layer}. In \bibinfo{booktitle}{\emph{Proceedings of the 30th ACM International Conference on Information \& Knowledge Management}}. \bibinfo{pages}{3568--3572}.
\newblock


\bibitem[Yao et~al\mbox{.}(2022)]%
        {yao2022pcl}
\bibfield{author}{\bibinfo{person}{Xufeng Yao}, \bibinfo{person}{Yang Bai}, \bibinfo{person}{Xinyun Zhang}, \bibinfo{person}{Yuechen Zhang}, \bibinfo{person}{Qi Sun}, \bibinfo{person}{Ran Chen}, \bibinfo{person}{Ruiyu Li}, {and} \bibinfo{person}{Bei Yu}.} \bibinfo{year}{2022}\natexlab{}.
\newblock \showarticletitle{PCL: Proxy-based Contrastive Learning for Domain Generalization}. In \bibinfo{booktitle}{\emph{Proceedings of the IEEE/CVF Conference on Computer Vision and Pattern Recognition}}. \bibinfo{pages}{7097--7107}.
\newblock


\bibitem[Yin et~al\mbox{.}(2012)]%
        {yin2012challenging}
\bibfield{author}{\bibinfo{person}{Hongzhi Yin}, \bibinfo{person}{Bin Cui}, \bibinfo{person}{Jing Li}, \bibinfo{person}{Junjie Yao}, {and} \bibinfo{person}{Chen Chen}.} \bibinfo{year}{2012}\natexlab{}.
\newblock \showarticletitle{Challenging the Long Tail Recommendation}.
\newblock \bibinfo{journal}{\emph{Proceedings of the VLDB Endowment 5.9}} (\bibinfo{year}{2012}).
\newblock


\bibitem[Zhang et~al\mbox{.}(2021)]%
        {zhang2021supporting}
\bibfield{author}{\bibinfo{person}{Dejiao Zhang}, \bibinfo{person}{Feng Nan}, \bibinfo{person}{Xiaokai Wei}, \bibinfo{person}{Shangwen Li}, \bibinfo{person}{Henghui Zhu}, \bibinfo{person}{Kathleen McKeown}, \bibinfo{person}{Ramesh Nallapati}, \bibinfo{person}{Andrew Arnold}, {and} \bibinfo{person}{Bing Xiang}.} \bibinfo{year}{2021}\natexlab{}.
\newblock \showarticletitle{Supporting Clustering with Contrastive Learning}.
\newblock \bibinfo{journal}{\emph{Proceedings of the 2021 Conference of the North American Chapter of the Association for Computational Linguistics: Human Language Technologies}} (\bibinfo{year}{2021}).
\newblock


\bibitem[Zhao et~al\mbox{.}(2020)]%
        {zhao2020memory}
\bibfield{author}{\bibinfo{person}{Xiangyu Zhao}, \bibinfo{person}{Haochen Liu}, \bibinfo{person}{Hui Liu}, \bibinfo{person}{Jiliang Tang}, \bibinfo{person}{Weiwei Guo}, \bibinfo{person}{Jun Shi}, \bibinfo{person}{Sida Wang}, \bibinfo{person}{Huiji Gao}, {and} \bibinfo{person}{Bo Long}.} \bibinfo{year}{2020}\natexlab{}.
\newblock \showarticletitle{Memory-efficient Embedding for Recommendations}.
\newblock \bibinfo{journal}{\emph{arXiv preprint arXiv:2006.14827}} (\bibinfo{year}{2020}).
\newblock


\bibitem[Zhou et~al\mbox{.}(2020)]%
        {zhou2020s3}
\bibfield{author}{\bibinfo{person}{Kun Zhou}, \bibinfo{person}{Hui Wang}, \bibinfo{person}{Wayne~Xin Zhao}, \bibinfo{person}{Yutao Zhu}, \bibinfo{person}{Sirui Wang}, \bibinfo{person}{Fuzheng Zhang}, \bibinfo{person}{Zhongyuan Wang}, {and} \bibinfo{person}{Ji-Rong Wen}.} \bibinfo{year}{2020}\natexlab{}.
\newblock \showarticletitle{S3-Rec: Self-supervised Learning for Sequential Recommendation with Mutual Information Maximization}. In \bibinfo{booktitle}{\emph{Proceedings of the 29th ACM International Conference on Information \& Knowledge Management}}. \bibinfo{pages}{1893--1902}.
\newblock


\bibitem[Zhu et~al\mbox{.}(2021)]%
        {zhu2021learning}
\bibfield{author}{\bibinfo{person}{Yongchun Zhu}, \bibinfo{person}{Ruobing Xie}, \bibinfo{person}{Fuzhen Zhuang}, \bibinfo{person}{Kaikai Ge}, \bibinfo{person}{Ying Sun}, \bibinfo{person}{Xu Zhang}, \bibinfo{person}{Leyu Lin}, {and} \bibinfo{person}{Juan Cao}.} \bibinfo{year}{2021}\natexlab{}.
\newblock \showarticletitle{Learning to Warm Up Cold Item Embeddings for Cold-start Recommendation with Meta Scaling and Shifting Networks}. In \bibinfo{booktitle}{\emph{Proceedings of the 44th International ACM SIGIR Conference on Research and Development in Information Retrieval}}. \bibinfo{pages}{1167--1176}.
\newblock


\end{thebibliography}


\clearpage
\appendix

\section{Hyper-parameter Tuning}

    As a regularization, the weight decay is tuned among \{0.0, 0.1, 0.2, 0.3\} in AdamW \cite{loshchilov2018fixing} optimizer with learning rate of 0.0001, and the random item replace technique from \cite{wu2019stochastic} is used with the probability tuned among \{0.0, 0.01, 0.02\}.
    The maximum length of input sequence mostly depends on the dataset; however, the performance gain converges when it reaches a certain value, so we used similar values reported in \cite{rashed2022context}.
    The other hyper-parameters are tuned as follows: dropout probability in \{0.0, 0.05, 0.1, 0.15, 0.2\}, normalization layers as either LayerNorm \cite{ba2016layer} or $L_{2}$-norm, the number of self-attention blocks $B$ in \{1, 2, 3, 4\}, the number of attention heads $H$ in \{1, 2, 4, 8\}, random sequence cut probability in \{0.0, 0.5, 1.0\}, temperature $\tau$ of NT-Xent in \{0.01, 0.1, 1.0\}, and the number of frequent items $K$ in \{0.0, 0.1, 0.25, 0.5, 0.75, 1.0\} times the number of total items.
    Note that most hyper-parameters are not sensitive when they reside on a certain value range.

\begin{figure}
    \centering
    \begin{subfigure}[b]{0.5\linewidth}
        \centering
        \includegraphics[width=1.0\linewidth]{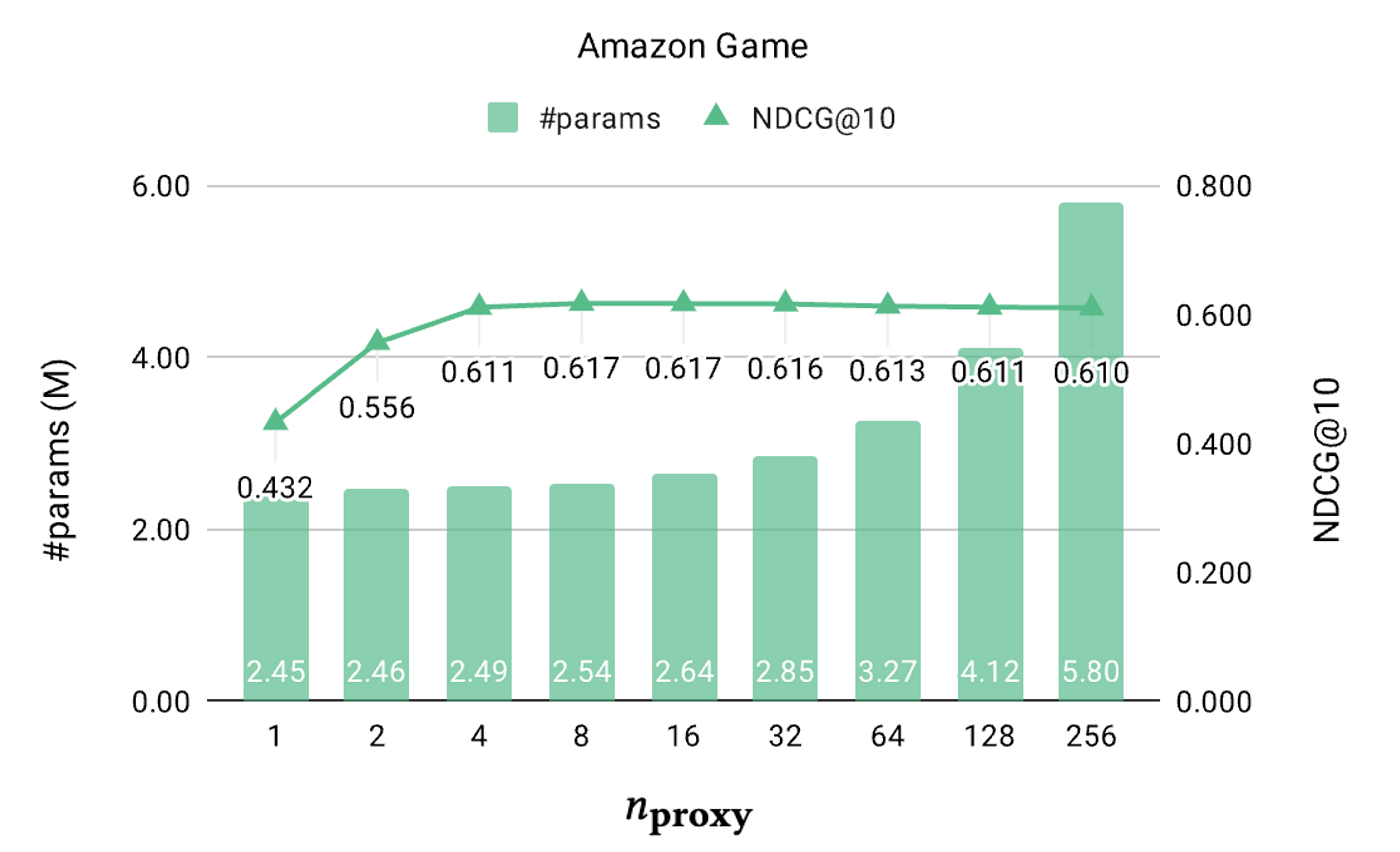}
        \caption{Impact of $n_{\text{proxy}}$.}
        \label{fig:ablation-proxy}
    \end{subfigure}%
    \begin{subfigure}[b]{0.5\linewidth}
        \centering
        \includegraphics[width=1.0\linewidth]{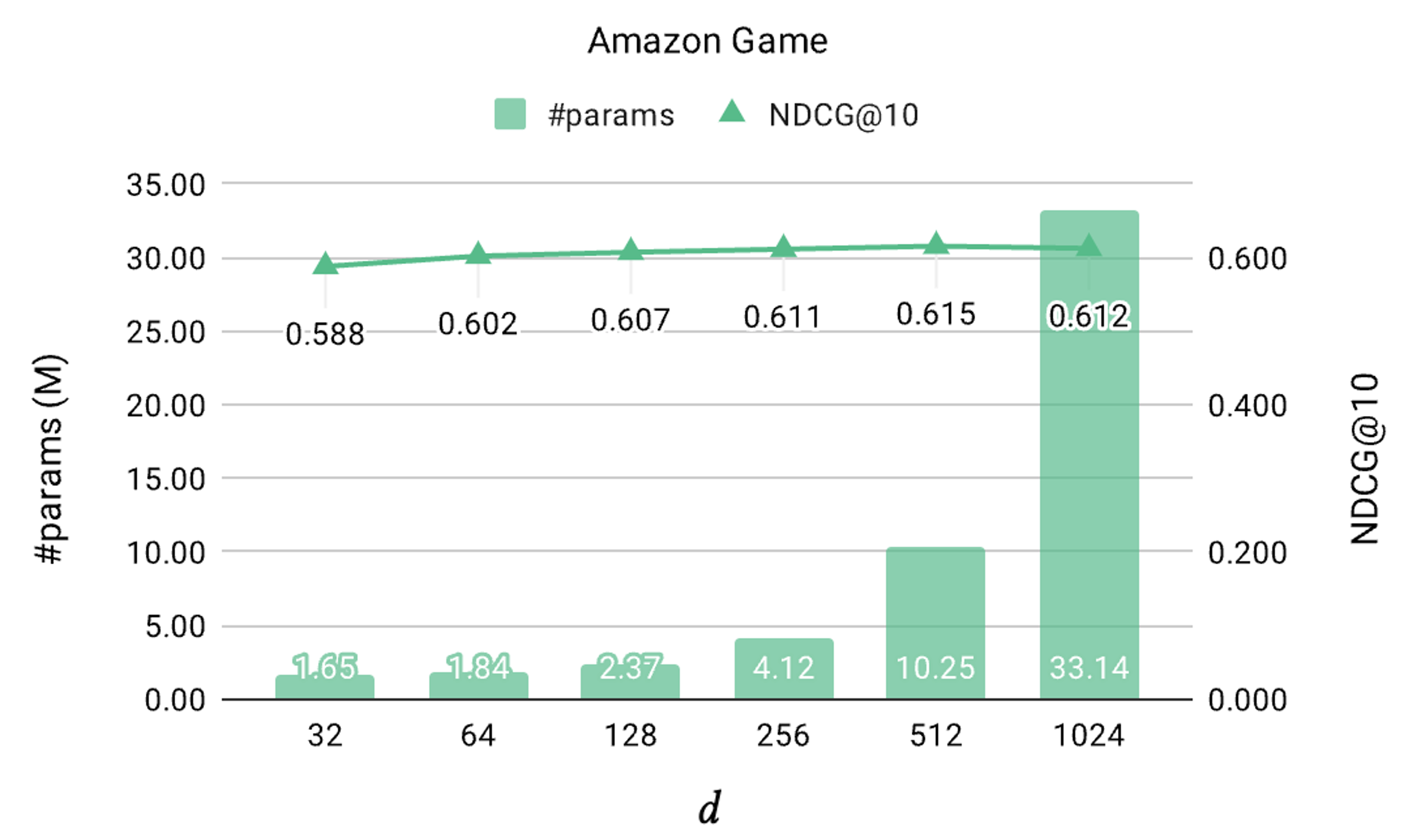}
        \caption{Impact of $d$.}
        \label{fig:ablation-hidden}
    \end{subfigure}
    \caption{Impact of hyper-parameters on Game.}
    \label{fig:ablation}
\end{figure}

\section{Ablation Study}

    The main hyper-parameters of PIR are $n_{\text{proxy}}$ and a latent dimension $d$, where these two primarily determines the number of parameters.
    As shown in Figure \ref{fig:ablation-proxy}, the model produces a similar performance once it reaches a certain number of proxies.
    Similarly, for the latent dimension, as demonstrated in Figure \ref{fig:ablation-hidden}, it shows that the model is surprisingly robust to the dimension.
    This suggests that besides increasing the model parameters in the recommendation model, there is still room for further exploration of the model structure to improve the performance \cite{ardalani2022understanding}.

\section{Properties of Proxy-based Item Representation}

    Here, we prove the representation propositions proposed in the main section.
    Since the model network is complex, we indirectly claim that the concept of content locality and bias priority are sound.

    \begin{proposition}{}
        The proxy-based item representation without bias term has content locality.
    \end{proposition}
    \begin{proof}
        Let $f^{(1)}, f^{(2)} \in \mathbb{R}^{n}$ be the content vectors of two items, right before the proxy weight computation.
        For simplicity, we assume that $\sum_{i=1}^{d} e^{f^{(1)}_{i}} = \sum_{i=1}^{d} e^{f^{(2)}_{i}} = Z$ and $f^{(1)}_{i} = f^{(2)}_{i}$ for $i = 3, 4, ..., d$, and $f^{(1)}_{1} = f^{(1)}_{2} = g$.
        Without loss of generality, $f^{(2)}_{1} = f^{(1)}_{1} + \epsilon$ and $f^{(2)}_{2} = f^{(1)}_{2} - \delta$ hold for some $\epsilon, \delta > 0$.
        Using the assumption, we can derive the $\delta$ as follow:
        \begin{align*}
            & \,\,\,\,
                \sum_{i=1}^{d} e^{f^{(1)}_{i}} = \sum_{i=1}^{d} e^{f^{(2)}_{i}} \\
            \Leftrightarrow & \,\,\,\,
                e^{f^{(1)}_{1}} + e^{f^{(1)}_{2}} = e^{f^{(1)}_{1} + \epsilon} + e^{f^{(1)}_{2} - \delta} \\
            \Leftrightarrow & \,\,\,\,
                2e^{g} = e^{g + \epsilon} + e^{g - \delta} \\
            \Leftrightarrow & \,\,\,\,
                \delta = -\log(2 - e^{\epsilon}), \\
        \end{align*}
        with conditions of $\epsilon < \log(2 - e^{-g})$.
        Now let $P \in \mathbb{R}^{n \times d}$ be proxy embeddings, then since the proxy weight is determined by the softmax of the content vectors, the item representation becomes $v = \text{softmax}(f)P$.
        Then the difference between two item representations is bounded above by using $\epsilon$:
        \begin{align*}
            & \,\,\,\,
                \left\| v^{(1)} - v^{(2)} \right\| \\
            = & \,\,\,\,
                \left\| \text{softmax}(f^{(2)})P - \text{softmax}(f^{(2)})P \right\| \\
            \leq & \,\,\,\,
                \left\| P \right\|
                \left\| \text{softmax}(f^{(2)}) - \text{softmax}(f^{(2)}) \right\|
                &&\text{(matrix norm)}
                \\
            = & \,\,\,\,
                \left\| P \right\|
                Z^{-1}
                e^{g}
                \left(
                    (1 - e^{\epsilon})^2
                    +
                    (1 - e^{-\delta})^2
                \right)^{\frac{1}{2}}
                \\
            = & \,\,\,\,
                \sqrt{2}
                \left\| P \right\|
                Z^{-1}
                e^{g}
                \left(
                    e^{2\epsilon} - 2e^{\epsilon} + 1
                \right)^{\frac{1}{2}}.
                \\
        \end{align*}
        Note that with Taylor expansion up to the second order, we can approximate $\left(e^{2\epsilon} - 2e^{\epsilon} + 1\right)^{\frac{1}{2}} \approx \epsilon$.
    \end{proof}

    \begin{proposition}{}
        The proxy-based item representation with bias term has bias priority.
    \end{proposition}
    \begin{proof}
        Even with the same content vector, the bias term can control the proxy weight if necessary.
        If two bias terms are $b^{(1)} = [\infty, 0, 0, ..., 0]$ and $b^{(1)} = [0, \infty, 0, ..., 0]$, then the weights are $w^{(1)} = \text{softmax}(f^{(1)} + b^{(1)}) = [1, 0, 0, ..., 0]$ and $w^{(2)} = \text{softmax}(f^{(2)} + b^{(2)}) = [0, 1, 0, ..., 0]$.
        Then, the final item representation is given as $w^{(1)}P = P_{1}$ and $w^{(2)}P = P_{2}$ respectively, which can denote completely different clusters.
    \end{proof}

\begin{table*}[tp]
    \caption{Performance comparison on Amazon datasets, including baseline models with NIP task. The best and the second best results are marked as bold and italic numbers respectively. The asterisk(*) denotes statistically significant (p < 0.05) gain against the non-PIR counterpart, using the $t$-test.}
    \label{tab:appendix-main}
    \begin{tabular}{cl|cccccccc}
        \toprule
            \multirow{2}{*}{Task} & \multirow{2}{*}{Model} & \multicolumn{2}{c}{Fashion} & \multicolumn{2}{c}{Men} & \multicolumn{2}{c}{Beauty} & \multicolumn{2}{c}{Game} \\
            & & R@10 & N@10 & R@10 & N@10 & R@10 & N@10 & R@10 & N@10 \\
        \midrule
            & Popular    & 0.407 & 0.262 & 0.415 & 0.269 & 0.451 & 0.261 & 0.519 & 0.314 \\
        \midrule
            \multirow{6}{*}{\begin{tabular}{c}NIP\end{tabular}}
            & BERT4Rec   & 0.328 & 0.209 & 0.315 & 0.193 & 0.478 & 0.318 & 0.705 & 0.509 \\
            & SASRec     & 0.381 & 0.245 & 0.397 & 0.259 & 0.485 & 0.322 & 0.742 & 0.541 \\
            & SSE-PT     & 0.381 & 0.246 & 0.397 & 0.258 & 0.502 & 0.337 & 0.775 & 0.566 \\
            & S$^{3}$Rec & 0.367 & 0.239 & 0.365 & 0.238 & 0.538 & 0.371 & 0.765 & 0.549 \\
            & SASRec++   & 0.546 & 0.344 & 0.500 & 0.315 & 0.545 & 0.351 & 0.752 & 0.533 \\
            & CARCA      & 0.591 & 0.381 & 0.550 & 0.349 & 0.579 & 0.396 & {\ul 0.782} & 0.573 \\
        \midrule
            \multirow{2}{*}{\begin{tabular}{c}BPR\end{tabular}}
            & BPR++
                & 0.523 & 0.332 & 0.429 & 0.266 & 0.505 & 0.352 & 0.768 & 0.564 \\
            & $\, \hookrightarrow$ with \textbf{PIR}
                & 0.620* & 0.406* & 0.554* & 0.358* & 0.511 & 0.353 & 0.770 & {\ul 0.582}* \\
        \midrule
            \multirow{5}{*}{\begin{tabular}{c}LIP\end{tabular}}
            & MixDim++
                & 0.623 & 0.407 & 0.570 & 0.365 & 0.587 & 0.398 & 0.766 & 0.556 \\
            & SASRec++
                & 0.630 & 0.416 & 0.587 & 0.379 & 0.601 & 0.415 & 0.753 & 0.539 \\
            & $\, \hookrightarrow$ with \textbf{PIR}
                & 0.635 & 0.426* & 0.580 & 0.381 & 0.599 & 0.422 & 0.779* & 0.572* \\
            & CARCA
                & {\ul 0.648} & {\ul 0.427} & {\ul 0.614} & {\ul 0.398} & {\ul 0.608} & {\ul 0.423} & 0.762 & 0.560 \\
            & $\, \hookrightarrow$ with \textbf{PIR} (\textbf{ProxyRCA})
                & \textbf{0.661}* & \textbf{0.446}* & \textbf{0.617} & \textbf{0.408} & \textbf{0.626}* & \textbf{0.449}* & \textbf{0.809}* & \textbf{0.611}* \\
        \midrule
        \midrule
            \multicolumn{2}{l}{Improvement over previous state-of-the-art} & 11.9\% & 17.0\% & 12.2\% & 16.9\% & 8.1\% & 13.4\% & 3.5\% & 6.6\% \\
        \bottomrule
    \end{tabular}
\end{table*}

\section{Additional Performance Comparison}

    \subsection{Comparison Models}

        Here we explain the all baseline models, including models with NIP task.
    
        \begin{enumerate}
            \item Popular: A non-personalized recommendation where the preference score is based on the item's global popularity. It serves as a sanity check for the performance lower bound.
            \item BERT4Rec \cite{sun2019bert4rec}: A sequential recommendation model using masked language model task from BERT with a bidirectional Transformer encoder.
            \item SASRec \cite{kang2018self}: A model that performs autoregressive sequential recommendation based on a unidirectional Transformer encoder.
            \item SSE-PT \cite{wu2020sse}: An extension of SASRec that strengthens personalization by applying the stochastic shared embedding technique on user embeddings.
            \item S$^3$Rec \cite{zhou2020s3}: An improved version of SASRec that adds various self-supervised tasks to leverage item attribute information further.
            \item SASRec++ \cite{rashed2022context}: An extension of SASRec \cite{kang2018self} that utilizes item attributes and context. We experimented two versions that uses NIP and LIP tasks respectively.
            \item CARCA \cite{rashed2022context}: A state-of-the-art ACSR model that uses cross-attention as an item scoring layer. Similar to SASRec++, we also compared two versions of NIP and LIP.
            \item BPR++ (ours): An extension of the non-sequential model BPR \cite{rendle2009bpr} that utilizes item attributes and context, with LIP task which uses multiple negative samples.
            \item MixDim++ (ours): An extension of CARCA with parameter-efficient mixed-dimension embedding \cite{ginart2021mixed} for encoding infrequent items. Note that this is an important baseline, which represents models that handle infrequent items (either long-tail or cold-start).
            \item ProxyRCA (ours): This is our main proposed model that employs proxy-based item representation as the item encoding layer on top of CARCA, named after \textbf{Proxy}-based item representation \textbf{R}ecommendation model with \textbf{C}ross-\textbf{A}ttention.
        \end{enumerate}

    \subsection{Overall Performance Comparison}

        Table \ref{tab:appendix-main} shows the overall performance comparison, including baseline models that use NIP task.
        The results show that LIP task alone can contribute to the performance improvement (by comparing NIP and LIP versions of CARCA), but when the task is fixed, PIR independently contributes to the performance.




\end{document}